\documentclass[a4paper,11pt]{article}
\usepackage{graphicx} 

\usepackage{amsthm,complexity,algorithm,algpseudocode}
\usepackage{amsmath}
\usepackage{amssymb}
\usepackage{geometry}
\usepackage{comment}
\usepackage[T1]{fontenc}
\usepackage[utf8]{inputenc}
\usepackage{authblk}
\usepackage{todonotes}
\usepackage{hyperref}
\usepackage{cleveref}
\usepackage{enumerate}
\usepackage{thm-restate}

\newif\iflong\longtrue

\newcommand{\cA}{{\mathcal A}}
\newcommand{\BB}{{\mathcal B}}

\newcommand{\OO}{{\mathcal O}}
\newcommand{\CC}{{\mathcal C}}

\newcommand{\GG}{{\mathcal G}}

\newcommand{\YY}{{\mathcal Y}}
\newcommand{\bN}{{\mathbb N}}
\newcommand{\ctvd}{{\sc Cliques or Trees Vertex Deletion}}
\newcommand{\VC}{{\sc Vertex Cover}}
\newcommand{\FVS}{{\sc Feedback Vertex Set}}
\newcommand{\CVD}{{\sc Cluster Vertex Deletion Set}}
\newcommand{\IVD}{{\sc Interval Vertex Deletion Set}}
\newcommand{\ChVD}{{\sc Chordal Vertex Deletion Set}}
\newcommand{\eps}{{\varepsilon}}

\newcommand{\defparproblem}[4]{
 \vspace{3mm}
\noindent\fbox{
  \begin{minipage}{.95\textwidth}
  \begin{tabular*}{\textwidth}{@{\extracolsep{\fill}}lr} \textsc{#1} \\ \end{tabular*}
  {\bf{Input:}} #2  \\
  {\bf{Parameter:}} #3\\
  {\bf{Question:}} #4
  \end{minipage}
  }
  \vspace{2mm}
}

\newtheorem{observation}{Observation}
\newtheorem{lemma}{Lemma}
\newtheorem{reduction rule}{Reduction Rule}
\newtheorem{definition}{Definition}
\newtheorem{proposition}{Proposition}

\title{A Polynomial Kernel for Deletion to the Scattered Class of Cliques and Trees \thanks{A preliminary version \cite{JacobMZ24} of this paper has appeared in proceedings of ISAAC-2024. Research of Diptapriyo Majumdar has been supported by Science and Engineering Research Board (SERB) grant SRG/2023/001592. Research of Meirav Zehavi has been supported by European Research Council (ERC) grant PARAPATH and Israel Science Foundation (ISF) grant no. 1470/24.}}
\author[1]{Ashwin Jacob \thanks{ashwinjacob@nitc.ac.in}}
\author[2]{Diptapriyo Majumdar\thanks{diptapriyo@iiitd.ac.in}}
\author[3]{Meirav Zehavi\thanks{meiravze@bgu.ac.il}}
\affil[1]{National Institute of Technology Calicut, Kozhikode, India}
\affil[2]{Indraprastha Institute of Information Technology Delhi, New Delhi, India}
\affil[3]{Ben-Gurion University of The Negev, Beersheba, Israel}

\date{}

\begin{document}

\maketitle

\begin{abstract}
The class of graph deletion problems has been extensively studied in theoretical computer science, particularly in the field of parameterized complexity.
Recently, a new notion of graph deletion problems was introduced, called {\em deletion to scattered graph classes}, where after deletion, each connected component of the graph should belong to at least one of the given graph classes. While fixed-parameter algorithms were given for a wide variety of problems, little progress has been made on the kernelization complexity of any of them.
 Here, we present the first non-trivial polynomial kernel for one such deletion problem, where, after deletion, each connected component should be a clique or a tree - that is, as densest as possible or as sparsest as possible (while being connected).  We develop a kernel of $\OO(k^5)$ vertices for the same.
\end{abstract}



\section{Introduction}
\label{sec:intro}

Graph modification problems form one of the most fundamental problem classes in algorithms and graph theory.
The input instance of a graph modification problem consists of an undirected/directed graph, and a non-negative integer $k$, and the objective is to decide if there exists a set of at most $k$ vertices/edges/non-edges whose deletion/addition yields in a graph belonging to some special graph class.
A specific graph modification problem allows to perform a specific graph operation, usually being {\em vertex deletion} or {\em edge deletion} or {\em edge addition} or {\em edge editing}.
Each of these operations, and vertex-deletion in particular, have been extensively studied from the perspective of classical and parameterized complexity.
For example, some vertex deletion graph problems that have received intense attention in the past three decades include {\VC}, {\FVS}, {\CVD}, {\IVD}, {\ChVD}, and more (see \cite{BoralCKP16,CaoM15,CaoM16,Chen16d,ChenKX10,DerbiszKMSSV22,Dumas023,fomin2013polynomial,KociumakaP14,JacobBDP21,RamanSS06}).

A graph class $\Pi$ is said to be {\em hereditary} if $\Pi$ is closed under induced subgraphs.
There are several hereditary graph classes $\Pi$ such that the corresponding {\sc $\Pi$-Vertex Deletion} problem is well-studied.
Such examples include all of the problems listed above.
Formally, the {\sc $\Pi$-Vertex Deletion} is defined as follows: given a graph $G$ and a non-negative integer $k$, we ask whether $G$ contains at most $k$ vertices whose deletion results in a graph belonging to class $\Pi$.
Lewis and Yannakakis \cite{LewisY80} proved that for every non-trivial $\Pi$, the {\sc $\Pi$-Vertex Deletion} problem is NP-complete.
Later, Cai \cite{Cai96} has proved that if a hereditary graph class $\Pi$ can be described by a finite set of forbidden subgraphs containing all minimal forbidden subgraphs in the class, then vertex deletion to $\Pi$ becomes fixed-parameter tractable (FPT) when solution size ($k$, the number of vertices that can be deleted) is considered as the parameter.

{Most of the computational problems that are {NP}-hard in general graphs can be solved in polynomial time when restricted to special graph classes.
For example, {\VC} is {NP-hard} on general graphs, but can be solved in polynomial time in forests, bipartite graphs, interval graphs, chordal graphs, claw-free graphs, and bounded treewidth graphs (see \cite{Courcelle90,CyganFKLMPPS15,golumbic2004algorithmic,minty1980maximal}).
Additionally, several other graph theoretic problems have also been studied in special graph classes (see \cite{BonamyDFJP19,BroersmaFGKPP15,CormenLRS09,EscoffierGM10,GolovachPL15,HeggernesKM09,JohnsonPP20,KlimosovaMMNPS20,KratschMT08,MartinPL20} for some examples).
{ If an input graph is such that each of its connected components belong to one of those special graph classes where a problem is polynomial time solvable, then the same problem can be solved in polynomial time by solving it over each component of the graph.}
Therefore, such a graph where each of the components belong to different graph classes are interesting.
We say that such graphs belong to a scattered graph class.
Vertex deletion problems are useful to find a set of few vertices whose removal results in a graph class where the problem of our interest is tractable.
Since the same problem is tractable in scattered graph classes (i.e. tractable in each of the graph class), vertex deletion to scattered graph classes are interesting to look at as well.}

{Many of the graph classes can be characterized by a set of forbidden graphs \cite{chudnovsky2006strong,Diestel-Book,golumbic2004algorithmic,lekkeikerker1962representation}. 
Vertex deletion problems for such graph classes boils down to hitting such forbidden subgraphs occuring as induced subgraphs of the input graph. 
Unlike this, for deletion to a scattered graph class, the deletion set $X$ might separate the vertices of the union of the forbidden subgraphs for each of the graph classes (instead of hitting them) so that all such graphs do not occur in any of the connected components of the graph $G-X$. This ramps up the difficulty for coming up with FPT, approximation and kernelization algorithms for deletion to scattered graph classes. A naive approach of finding the solutions (or kernels) for each of the deletion problems separately and ``combining'' them is unlikely to work.}

Ganian et al. \cite{GanianRS17} studied backdoors to scattered classes of CSP problems.
Subsequently, Jacob et al. \cite{JacobKMR23,JacobMR23jcss} built on the works by Ganian et al. \cite{GanianRS17} and initiated a systematic study of vertex deletion to scattered graph classes.
They considered the {\sc ($\Pi_1,\ldots,\Pi_d$)-Deletion} problem where the input instance is a graph $G$ a parameter $k$ with respect to $d$ (constant) hereditary graph classes $\Pi_1,\ldots,\Pi_d$.
The objective is to decide if there is a set of at most $k$ vertices $S$ such that every connected component of $G - S$ is in $\Pi_i$ for some $i \in [d]$.
After that, Jacob et al. \cite{JacobMR23jcss} considered specific pairs of hereditary graph classes $\Pi_1$ and $\Pi_2$ and have provided { single-exponential time} fixed-parameter tractable (FPT) algorithms and approximation algorithms for {\sc $(\Pi_1, \Pi_2)$-Deletion} problems.
Very recently, Jansen et al. \cite{JansenK023,JansenKW25} conducted a follow-up work on {\sc ($\Pi_1,\ldots,\Pi_d$)-Deletion} problems and have improved the results appearing in \cite{JacobKMR23}.
{A common theme for the FPT algorithms for deletion to scattered graph classes is a  non-trivial ``unification'' of the techniques used in the deletion problems of each of the graph classes.}

\paragraph{Our Problem and Results:} To the best of our knowledge, vertex deletion to scattered graph classes is essentially unexplored from the perspective of polynomial kernelization that is a central subfield of parameterized complexity.
The only folklore result that follows from Jacob et al.~\cite{JacobMR23jcss} states that if there are two hereditary graph classes $\Pi_1$ and $\Pi_2$ such that both $\Pi_1$ and $\Pi_2$ can be described by finite forbidden families and $P_d$ (the induced path of $d$ vertices) is a forbidden induced subgraph for $\Pi_1$ for some fixed constant $d$, then the problem {\sc $(\Pi_1, \Pi_2)$-Deletion} can be formulated as a {\sc $d$-Hitting Set} problem and hence admits a polynomial kernel (see \cite{Faisal2010,NeidermeierR2003}).
This folklore result is very restrictive and does not capture any hereditary graph class whose forbidden sets are not bounded by a fixed constant.


In this paper, we initiate the study of vertex deletion to scattered graph classes from the perspective of polynomial kernelization.
We consider the problem {\ctvd} where given a graph $G$ and a non-negative integer $k$, we ask if $G$ contains a set $S$ of at most $k$ vertices, such that $G - S$ is a simple graph and every connected component of it is either a clique or a tree -- that is, as densest as possible or as sparsest as possible (while being connected).
Naturally, we are specifically interested in the case where the input graph is already a simple graph.
However, our preprocessing algorithm can produce intermediate multigraphs.
Hence, we directly consider this more general formulation.
Formally, we define our problem as follows.

\defparproblem{{\ctvd} (CTVD)}{An undirected (multi)graph $G = (V, E)$ and a non-negative integer $k$.}{$k$}{Does $G$ contain a set $S$ of at most $k$ vertices such that $G -S$ is a simple graph and every connected component of $G - S$ is either a clique or a tree?}

This problem is particularly noteworthy as it captures the essence of scattered classes: allowing the connected components to belong to vastly different graph classes and ideally the simplest ones where various computational problems are polynomial time solvable.
Here, we indeed consider the extremes: the simplest densest graph (cliques) and the most natural class of sparsest connected graphs (trees).
If $X$ is a feasible solution to {\ctvd} for a graph $G$, then we call $X$ a {\em (clique, tree)-deletion set} of $G$.
We consider the (upper bound on the) solution size $k$ as the most natural parameter.
Jacob et al. \cite{JacobMR23jcss} proved that {\ctvd} is in FPT ---specifically, that it admits an algorithm that runs in $\OO^*(4^k)$ time.
We prove the following result on polynomial kernelization for this problem.

\begin{restatable}{theorem}{MainResult}
\label{thm:main-result}
{\ctvd} (CTVD) admits a kernel with $\OO(k^5)$ vertices.	
\end{restatable}

Our theorem is the first non-trivial result on a polynomial kernel for vertex deletion to pairs of graph classes.
{The proof of this kernelization upper bound is based on several non-trivial insights, problem specific reduction rules, and structural properties of the solutions}.

\paragraph{Remark.} Very recently, Tsur \cite{TSUR2025Kernel} proved that {\ctvd} admits a kernel with $\OO(k^4)$ vertices and it was subsequently improved to a kernel with $\OO(k^2)$ vertices by Kumabe \cite{Kumabe2025}.
%
\paragraph{Organization of the Paper:} In Section \ref{sec:prelims}, we introduce basic terminologies and notations.
Section \ref{sec:clique-tree-kernel} is devoted to the proof of our main result.
Finally, in Section \ref{sec:conclusion}, we conclude with some future research directions.

\section{Preliminaries}
\label{sec:prelims}

\paragraph{Graph Theory:}
We use standard graph theoretic terminologies from Diestel's book \cite{Diestel-Book}. 
For a vertex $v$ in $G$, let $d_G(v)$ denote the degree of $v$ in $G$, which is the number of edges in $G$ incident to $v$. When we look at the number of edges incident on $v$, we take the multiplicity of every edge into account.
A {\em pendant vertex} in a graph $G$ is a vertex having degree one in $G$.
A {\em pendant edge} in a graph $G$ is an edge incident to a pendant vertex in $G$.
A {\em path} $P$ in a graph is a sequence of distinct vertices $(v_1,\ldots,v_r)$ such that for every $1 \leq i \leq r-1$, $v_i v_{i+1}$ is an edge.
A {\em cycle} in a graph is a sequence of vertices $(v_1,v_2,\ldots,v_r,v_1)$ such that $r \geq 2$, $v_1,\ldots,v_r$ are distinct, $v_r v_1$ is an edge, and for every $1 \leq i \leq r-1$, $v_i v_{i+1}$ is an edge.
{
The {\em length} of a path (respectively, a cycle) is defined as the number of edges in the path.
A {\em degree-2-path} in $G$ is a path $P$ such that all its internal vertices have degree exactly 2 in $G$.
If a graph $G$ has a degree-2-path $P = (v_1,\ldots,v_{r})$ such that $v_1$ is a pendant vertex, for every $i \in \{2,\ldots,r-1\}$, $d_G(v_i) = 2$ and $d_G(v_{r}) > 2$, then we call $P$ a {\em degree-2-tail} of $r$ vertices.
If $G$ has a degree-2-path $P = (v_1,\ldots,v_r)$ such that for all $i \in \{2,\ldots,r-1\}$, $d_G(v_i) = 2$, and $d_G(v_1), d_G(v_r) > 2$, then we call $P$ a {\em degree-2-overbridge} of $r$ vertices.
An induced cycle $(v_1, v_2, \ldots, v_{\ell}, v_1)$ with $\ell \geq 4$ is a {\em pendant hole} attached to $v_1$ if for every $2 \leq i \leq \ell-1$, $v_i$ has no neighbor other than $v_{i-1}$ and $v_{i+1}$, and $v_{\ell}$ is adjacent to only $v_1$ and $v_{\ell -1}$.}
Sometimes, for simplicity, we use $P = v_1-v_2-\ldots-v_r$ to denote the same path $P$ of $r$ vertices.
The graph $K_t$ for integer $t \geq 1$ is the clique of $t$ vertices.
For a non-negative integer $c$, the graph $cK_1$ is the collection of $c$ isolated vertices.
The graphs $C_t$ and $P_t$ for integer $t \geq 1$ are the cycle and path of $t$ vertices respectively.
We define a {\em paw graph} as the graph with four vertices $u_1, u_2, u_3$ and $u_4$ where $u_1, u_2, u_3$ form a triangle, and $u_1$ alone is adjacent to $u_4$ (thus, $u_4$ is a pendant vertex). We define a {\em diamond graph} as the graph with four vertices $u_1, u_2, u_3$ and $u_4$ where $u_1, u_2, u_3$ form a triangle, and $u_1, u_2$ are adjacent to $u_4$. Note that both paw and diamond contain a triangle as well as a $2K_1$ as induced subgraphs.
We say that a vertex $x \in V(G)$ is {\em adjacent} to a subgraph $G[Y]$ for some $x \notin Y$ if $x$ has a neighbor in $Y$ in the graph $G$.
A {\em cut} of $G$ is a bipartition $(X, Y)$ of $V(G)$ into nonempty subsets $X$ and $Y$.
The set $E_G(X, Y)$ is denoted as the edges {\em crossing the cut}.
We omit the subscript when the graph is clear from the context.
Let $C$ be a cycle having $r$ vertices.
We use $C = v_1-v_2-\ldots-v_r-v_1$ to denote the cycle with edges $v_i v_{i+1}$ for every $1 \leq i \leq r-1$ and the edge $v_r v_1$. 

\begin{definition}[$v$-flower]
\label{defn:flower}
For a graph $G$ and a vertex $v$ in $G$, a {\em $v$-flower}  is the structure formed by a family of $\ell$ cycles $C_1, C_2, \dots C_{\ell}$ in $G$ all containing $v$ and no two distinct cycles $C_i$ and $C_j$ sharing any vertex except $v$.
We refer to the $C_i$s' as the {\em petals} and  to $v$ as the {\em core}.
The number of cycles $\ell$ is the {\em order} of the $v$-flower. 
\end{definition}

\begin{proposition}[\cite{CyganFKLMPPS15}, Lemma 9.6]
\label{theorem:flower-cycle-hitting}
 Given a graph $G$ with $v \in V (G)$ and an integer $k$, there exists a polynomial time algorithm that either provides a $v$-flower of order $k+1$ or compute a set $Z \subseteq V (G) \setminus \{v\}$ with at most $2k$ vertices that intersects every cycle of $G$ that passes through $v$.
\end{proposition}

Let $q$ be a positive integer and $G$ be a bipartite graph with vertex bipartition $(A, B)$.
For $\widehat A \subseteq A$ and $\widehat B \subseteq B$, a set $M \subseteq E(G)$ of edges is called a {\em $q$-expansion} of $\widehat A$ into $\widehat B$ if
\begin{enumerate}[(i)]
	\item every vertex of $\widehat A$ is incident to exactly $q$ edges in $M$, and
	\item exactly $q|A|$ vertices of $\widehat B$ are incident to the edges in $M$.
\end{enumerate}
The vertices of $\widehat A$ and $\widehat B$ that are the endpoints of the edges of $M$ are said to be {\em saturated} by $M$.
As a remark, it is important to note that by definition of $q$-expansion $M$ of $\widehat A$ into $\widehat B$, all vertices of $\widehat A$ are saturated by $M$, and $|\widehat B| \geq q|\widehat A|$. But not all vertices of $\widehat B$ are guaranteed to be saturated by $M$.
 
\begin{lemma}[$q$-Expansion Lemma \cite{CyganFKLMPPS15,Thomasse10}]
\label{lemma:expansion}
{ Let $q \in \bN$ and $G$ be a bipartite graph with $n$ vertices and $m$ edge.
 Suppose that the vertex bipartition of $G$ is $(A, B)$ such that $|B| \geq q|A|$, and there is no isolated vertex in $B$.}
Then, there exist non-empty vertex sets $X \subseteq A$ and $Y \subseteq B$ such that
\begin{enumerate}[(i)]
	\item there is a $q$-expansion $M$ of $X$ into $Y$, and
	\item no vertex in $Y$ has a neighbor outside $X$, that is, $N(Y) \subseteq X$.
\end{enumerate}
Furthermore, the sets $X$ and $Y$ can be found in time $\OO(mn^{1.5})$.
\end{lemma}

Recently, Fomin et al. \cite{FominLLSTZ19} have designed the following generalization of the Lemma \ref{lemma:expansion} ($q$-Expansion Lemma) as follows.

\begin{lemma}[New $q$-Expansion Lemma \cite{BabuJKR24,FominLLSTZ19,JacobMR23algo}]
\label{lemma:new-expansion-lemma}
Let $q$ be a positive integer and $G$ be a bipartite graph with bipartition $(A, B)$.
Then there exists $\widehat A \subseteq A$ and $\widehat B \subseteq B$ such that there is a $q$-expansion $M$ of $\widehat A$ into $\widehat B$ in $G$ such that
\begin{enumerate}[(i)]
	\item $N(\widehat B) \subseteq \widehat A$, and
	\item $|B \setminus \widehat B| \leq q|A \setminus \widehat A|$.
\end{enumerate}
Furthermore, the sets $\widehat A$, $\widehat B$ and the $q$-expansion $M$ can be computed in polynomial time.
\end{lemma}

{Observe that the { statement of Lemma \ref{lemma:new-expansion-lemma}} does not require the condition that $B$ has no isolated vertex and $|B| \geq q|A|$ that were required for the Lemma \ref{lemma:expansion}.}
In particular, if $|B| > q|A|$, then it must be that $|\widehat B| > q|\widehat A|$ and $\widehat B$ will contain some vertex that is not saturated by the $q$-expansion $M$.


\paragraph{Forbidden Subgraph Characterization:}
Given a graph class $\GG$, {any (induced) subgraph that is not allowed to appear in any graph of $\GG$ is called an {\em obstruction} for $\GG$} (also known as {forbidden subgraphs} or forbidden induced subgraphs).
{
An obstruction $O$ is a {\em minimal} obstruction to the graph class $\GG$ if there does not exist any obstruction $Q \neq O$ that is an induced subgraph of $O$.
For our problem {\ctvd}, we consider the set of all minimal obstructions.
For the ease of presentation, we use ``obstructions to {\ctvd}'' denote ``the set of all minimal obstructions to the class of graphs that is a disjoint union of cliques and trees.}
We first identify the obstructions for {\ctvd}.
Clearly, on simple graphs, we cannot have an obstruction for both a tree and a clique in the same connected component.
If $\GG$ is the class of all cliques, then the obstruction for $\GG$ is $2K_1$ and if $\GG$ is the class of all forests, then any cycle $C_t$ with $t \geq 3$ is an obstruction for $\GG$.
The obstructions for both a clique is the graph $2K_1$, and for trees are the cycles $C_t$ with $t \geq 3$.
Note that a cycle $C_t$ with $t \geq 4$ contains $2K_1$ as an induced subgraph.
Throughout the paper, we sometimes abuse the notation where an obstruction (or a forbidden induced subgraph) is viewed as a set and sometimes it is viewed as an (induced) subgraph.

 \begin{observation}
 \label{obs:Ct-obstruction-clique-tree}
 For every integer $t \geq 4$, the cycles $C_t$ contains $2K_1$ as induced subgraph. 
 \end{observation}
 
Thus, we can conclude that the obstructions for {\ctvd} are cycles $C_t$ with $t \geq 4$ and connected graphs with both $2K_1$ and $C_3$ as induced subgraphs.
For multigraphs, a vertex with a self-loop and two vertices with two (or more) edges are obstructions as well. 
If a connected graph has both $2K_1$ and $C_3$ as induced subgraphs, a (clique, tree)-deletion set  either intersects the union of the vertex sets of these subgraphs or contains a subset separating them.
The following lemma claims that if a connected graph contains both $2K_1$ and  $C_3$ as induced subgraphs, then it contains a paw or a diamond.  

\begin{lemma}
\label{lemma:forbidden-2k1-c3}
A connected graph $G$ with both $2K_1$ and $C_3$ as induced subgraphs contains either a paw or a diamond as an induced subgraph.
\end{lemma}

\begin{proof}
Since $G$ is a connected graph that is not a clique and contains a triangle, there are three vertices $u_1, u_2, u_3 \in V(G)$ and a vertex $w \notin U$ such that $U = \{u_1,u_2,u_3\}$ induces a triangle, $w$ is adjacent to a vertex in $U$ and not adjacent to a vertex in $U$. Suppose, without loss of generality that $w$ is adjacent to $u_1$ and not adjacent to $u_2$.
If $w$ is not adjacent to $u_3$, the graph induced by $ \{u_1, u_2, u_3,w\}$ forms a paw. Otherwise, it forms a diamond.
\end{proof}

From Lemma \ref{lemma:forbidden-2k1-c3}, we get a forbidden subgraph characterization for the class of graphs where each connected component is a clique or a tree.

\begin{lemma}
\label{ctvd-forbidden-char}
Let $\GG$ be the class of all simple graphs where each connected component is a clique or a tree.
Then, a simple graph $G$ belongs to $\GG$ if and only if $G$ does not contain any paw, diamond or cycle $C_i$ with $i \geq 4$ as an induced subgraph.
\end{lemma}

\begin{proof}
We first give the forward direction ($\Rightarrow$) of the proof.
Let a simple graph $G$ { be} in ${\GG}$.
The obstruction for cliques is the graph $2K_1$, and the obstructions for trees are cycles $C_t$ with $t \geq 3$.
Due to the Observation \ref{obs:Ct-obstruction-clique-tree}, the cycles $C_t$ with $t \geq 4$ contain $2K_1$ as an induced subgraph. 
Thus, we can conclude that the obstructions for ${\GG}$ include all cycles $C_t$ with $t \geq 4$.
It follows from Lemma \ref{lemma:forbidden-2k1-c3} that all connected graphs with both $2K_1$ and  $C_3$ as induced subgraphs contain a paw or a diamond as an induced subgraph.
 Thus, the obstructions for $\GG$ include  paws, diamonds and cycles $C_i$ with $i \geq 4$ as induced subgraphs.
 This completes the proof that $G$ does not contain any paw, diamond or cycle $C_i$ with $i \geq 4$ as induced subgraph.
 
Now, we give the backward direction ($\Leftarrow$) of the proof.
Let $G$ be a simple graph that does not contain any paw, diamond or $C_i$ for some $i \geq 4$ as induced subgraphs but $G \notin \GG$.
Then, $G$ contains a connected component $K$ that is neither a clique nor a tree.
If $K$ contains an induced cycle with more than three vertices, then this leads to a contradiction.
Hence, we can assume that $K$ contains a pair of nonadjacent vertices as well as a triangle (i.e. $C_3$).
It follows from Lemma \ref{lemma:forbidden-2k1-c3} that $K$ contains either a paw or a diamond as an induced subgraph.
This also leads to a contradiction completing the proof of the lemma.
\end{proof}

\paragraph{Parameterized Complexity and Kernelization:}
A {\em parameterized problem} $L$ is a set of instances $(x, k) \in \Sigma^* \times \bN$ where $\Sigma$ is a finite alphabet and $k \in \bN$ is a parameter.
The notion of `tractability' in parameterized complexity is defined as follows.

\begin{definition}[Fixed-Parameter Tractability]
\label{def:FPT}
A parameterized problem $L$ is {\em fixed-parameter tractable} (or FPT) if there exists an algorithm $\cA$ that on input $(x, k) \in \Sigma^* \times \bN$, runs in $f(k)|x|^{\OO(1)}$ time for some computable function $f: \bN \rightarrow \bN$ and correctly decides if $(x, k) \in L$. 
This algorithm $\cA$ is called {\em fixed-parameter algorithm} (or {\em FPT algorithm}) for the problem $L$.	
\end{definition}

We say that two instances $(x, k)$ of $L$ and $(x', k')$ of $L$ are {\em equivalent} if $(x, k) \in L$ if and only if $(x', k') \in L$.
The notion of kernelization for in the area of parameterized complexity is defined as follows.

\begin{definition}[Kernelization]
\label{defn:kernel}
A {\em kernelization} for a parameterized problem $L \subseteq \Sigma^* \times \bN$ is an algorithm that given an instance $(x, k)$ of $L$, outputs an equivalent instance $(x', k')$ (called {\em kernel}) of $L$ in time polynomial in $|x| + k$ such that $|x'| + k' \leq g(k)$ for some computable function $g: \bN \rightarrow \bN$.
If $g(k)$ is $k^{\OO(1)}$, then $L$ admits a {\em polynomial kernel}.
\end{definition}

A kernelization algorithm usually consists of a collection of {\em reduction rules} that have to be applied exhaustively.
A reduction rule is {\em safe} if given an instance $(x, k)$ of $L$, one application of the reduction rule outputs an equivalent instance $(x', k')$ of $L$.
It is well-known due to \cite{CyganFKLMPPS15,DowneyFellowsBook13,FominLLSTZ19} that a decidable parameterized problem is FPT if and only if it admits a kernelization.
We refer to \cite{CyganFKLMPPS15,DowneyFellowsBook13,FominLLSTZ19} for more formal details related to parameterized complexity and kernelization.

\section{Polynomial Kernel for Cliques and Trees}
\label{sec:clique-tree-kernel}

This section is devoted to the proof of our main result (Theorem \ref{thm:main-result}) which is a polynomial kernel with $\OO(k^5)$ vertices for {\ctvd}.
As the first step, we invoke the following proposition by Jacob et al. \cite{JacobMR23jcss} that computes a (clique, tree)-deletion set $S \subseteq V(G)$ with at most $4k$ vertices.

\begin{proposition}[\cite{JacobMR23jcss}, Theorem 6]
\label{lemma-ctvd-approximation}  {\ctvd} admits a $4$-approximation algorithm.
\end{proposition}

We begin with the following observation, whose proof is trivial.

\begin{observation}
\label{observation:subgraph-yes-instance}
The following statements hold true.
\begin{enumerate}
	\item For any subset $Z \subseteq V(G)$, if $(G,k)$ is a yes-instance for {\ctvd} with solution $X$, then $(G-Z,k)$ is a yes-instance for {\ctvd} with solution $X \setminus Z$. 
	\item { Suppose that $Z \subseteq V(G)$ is a vertex subset that is disjoint from every induced subgraph isomorphic to an obstruction for CTVD and $(G - Z, k)$ is a yes-instance to {\ctvd} with solution $X$ of size at most $k$.
	Then, $X$ is a solution of size at most $k$ to the instance $(G, k)$.}
\end{enumerate}
\end{observation}

\paragraph{Overview of the Kernelization Algorithm:}
Our first step is to invoke Proposition \ref{lemma-ctvd-approximation} to get a set $S$ such that $|S| \leq 4k$ (as otherwise $(G, k)$ is a no-instance).
In Section \ref{sec:initial-preprocessing-rules}, we provide some reduction rules that ensures us that every connected component of $G - S$ has a neighbor in $S$, and the graph has no degree-2-path of $G$ has more than 4 vertices.
Subsequently, in Section \ref{sec:clique-vertex-bounding}, we provide some reduction rules and prove that the number of vertices in the connected components of $G - S$ that are cliques is $\OO(k^5)$.
Finally, in Section \ref{sec:tree-vertices-G-S}, we prove a variant of Proposition \ref{theorem:flower-cycle-hitting} and use reduction rules related to that.
In addition, we use New $q$-Expansion Lemma (i.e. Lemma \ref{lemma:new-expansion-lemma}) to develop further reduction rules and structural properties of the graph to prove that the number of vertices in the connected components of $G - S$ that are trees is $\OO(k^2)$. 

\subsection{Initial Preprocessing Rules}
\label{sec:initial-preprocessing-rules}

Let $S$ be a $4$-approximate (clique, tree)-deletion set of $G$ obtained from Proposition \ref{lemma-ctvd-approximation}.
If $|S| > 4k$, we conclude that $(G,k)$ is a no-instance and return a trivial constant sized no-instance.
Hence, we can assume without loss of generality that $|S| \leq 4k$. 
We also assume without loss of generality that $G$ has no connected component that is a clique or a tree.
If such components are there, we can delete those components.
Hence, we can naturally assume from now onwards, that every connected component of $G - S$ (that is either a clique or a tree) has some neighbor in $S$.
{ But some of our subsequent reduction rules can create some component in $G$ that is a clique or a tree}.
So, we state the following reduction rule for the sake of completeness.
{ In this following reduction rule (whose safeness easily follows), we delete every connected component $C$ in $G - S$ that has no neighbor in $S$}.
All the obstructions for {\ctvd} are connected graphs and intersect with $S$. 
Thus, no obstruction can be part of isolated component $C$, and also $S$.

\begin{reduction rule}
\label{rr1}
{ If a connected component $C$ of $G$ is a clique or a tree, then remove $C$ from $G$.}
The new instance is $(G - C, k)$.
\end{reduction rule}

Since one of our subsequent reduction rules can create parallel edges which is also an obstruction, we state the following reduction rule whose safeness is also trivial.

\begin{reduction rule}
\label{rule:multiplicity-reduction}
If there is an edge with multiplicity more than two, reduce the multiplicity of that edge to exactly two.	
\end{reduction rule}

{%
As we consider the more generalized formulation where the input is allowed to be a multigraph and double-edges are also among the obstruction sets, we need the following reduction rule that is similar to the Buss-Rule analog for the {\VC} problem (see \cite{CyganFKLMPPS15,DowneyFellowsBook13}).

\begin{reduction rule}
\label{rule:high-degree-double-edge}
If there is $u \in V(G)$ that has $k+1$ adjacent vertices $v_1,\ldots,v_{k+1}$ such that for every $i \in [k+1]$, there is a double-edge between $u$ and $v_i$, then the new instance is $(G - u, k-1)$.
\end{reduction rule}

The safeness of the above reduction rule is obvious because double-edges are obstructions.
Clearly, if $u$ is not picked in a solution, then at least $k+1$ vertices become essential.
Hence, any solution of size at most $k$ must contain $u$.}

We note that any induced cycle of length at least 4 is a cycle with no chord and has at least 4 vertices.
From now onwards, for the sake of simplicity, we will often use the term {\em chordless cycle} to mean a cycle $C_{\ell}$ with $\ell \geq 4$ as induced subgraphs.
We also have the following reduction rule that helps us bounding the number of pendant vertices attached to any vertex.

\begin{reduction rule}
\label{rr4}
If there exists a vertex $u$ in $G$ adjacent to vertices $v$ and $v'$ that are pendants in $G$, then remove $v$ from $G$. 
The new instance is $(G - v, k)$.
\end{reduction rule}

\begin{lemma}
\label{lemma:safeness-pendant-vertex-deletion-rule}
Reduction Rule \ref{rr4} is safe.
\end{lemma}

\begin{proof}
From Observation \ref{observation:subgraph-yes-instance}, if $(G, k)$ is a yes-instance, then $(G -v, k)$ is also a yes-instance.
For the backward direction ($\Leftarrow$), let $X$ be a solution to $(G-v, k)$.
Targeting a contradiction, suppose that $X$ is not a feasible solution of $G$.
First of all observe that $v$ is a pendant, therefore, there cannot be a double-edge containing $v$.
As $G - (X \cup \{v\})$ is a simple graph, and $v$ is a pendant vertex hence $G - X$ is also a simple graph.
So, we can assume that $G - X$ is a simple graph but has a connected component that is neither a clique nor a tree.
Then by Lemma \ref{ctvd-forbidden-char}, $G - X$ contains a paw, or a diamond, or a cycle $C_i$ with $i \geq 4$ as an induced subgraph.
Note that this induced subgraph is not present in $G-(X \cup \{v\})$ as $X$ is a solution in $G-v$.
Thus one of the vertices of this induced subgraph has to be $v$.
Since $v$ is a degree-1 vertex in $G$, the induced subgraph has to be a paw that contains $v$ and its unique neighbor $u$.
Let the vertices of this paw of $G - X$ be $u,v,w_1$ and $w_2$.
Note that $u, v, w_1, w_2 \notin X$.
Since $v'$ is also a degree-1 vertex adjacent to $u$, the graph induced by  $u,v',w_1$ and $w_2$ is also a paw in $G$.
As $u, v', w_1, w_2$ { form} a paw in $G - v$, $X$ is a solution to $G - v$, and $u, w_1, w_2 \notin X$, it must be that $v' \in X$.
As $v' \in X$, we set $X' = (X \setminus \{v'\}) \cup  \{u\}$.
We claim that $X'$ is also a feasible solution to the graph $G$.
 
One crucial observation is that the collection of all chordless cycles, double-edges, and diamonds of $G$ and $G - v$ are the same.
Observe that all diamonds, double-edges, and chordless { cycles of} $G - v$ are already intersected by $X$.
Hence, $X$ also intersects all chordless cycles, double-edges, and diamonds of $G$.
Also, $X'$ contains all vertices of $X$ except $v'$.
As $v'$ is a pendant vertex in $G$, $X'$ intersects all double-edges, all chordless cycles, and all diamonds of $G$.
As $v'$ is not participating in any chordless cycle, double-edge, and diamond of $G - v$, the objective of $v'$ is to only intersect the paws of $G - v$ that contain $v'$.
Observe that any such paw of $G - v$ will also contain $u$.
As $u \in X'$ and $X' = (X \setminus \{v'\}) \cup \{u\}$, hence $X'$ intersects every paw of $G - v$ containing $v'$.
Thus $X'$ is a solution to $G - v$.
Now, we argue that $X'$ intersects all paws, diamond, double-edges, and chordless cycles of $G$.
First of all, as $v$ is a pendant vertex, $v$ does not participate in any chordless cycle, double-edge, or diamond of $G$.
We have already argued that all double-edges, all diamonds, and all chordless cycles of $G$ are intersected by $X'$.
As the set of double-edges of $G$ and $G - v$ are the same, and any paw of $G$ that contains $v$ also contains its unique neighbor $u$.
As $u \in X'$, it follows that $X'$ intersects all paws containing $v$.
Therefore, $X'$ is a solution to $G$.
As $|X'| = |X|$, this completes the proof of safeness of this reduction rule.
\end{proof}

We can conclude that if Reduction Rule \ref{rr4} is not applicable, then every vertex in $G$ is adjacent to at most one pendant vertex in $G$.
From now, we assume that every vertex in $G$ is adjacent to at most one pendant vertex.
Our next two reduction rules help us to reduce the length (the number of vertices) of a degree-2-path in $G$. 
Note that a degree-2-path can be of two types, either a degree-2-tail or a degree-2-overbridge.
For both types, we need the following two reduction rules.

\begin{reduction rule}
\label{rr5}
Let $P = (v_1, v_2, \ldots, v_{\ell})$ be degree-2-tail { with $\ell$ vertices} such that $d_G(v_1) > 2$, $d_G(v_{\ell}) =1$ and $Z = \{v_3, v_4, \dots, v_{\ell}\}$.
Then, remove $Z$ from $G$. The new instance is $(G-Z, k)$.  
\end{reduction rule}

\begin{lemma}
Reduction Rule \ref{rr5} is safe.
\end{lemma}


\begin{proof}
From Observation \ref{observation:subgraph-yes-instance}, if $(G, k)$ is a yes-instance, then $(G -Z, k)$ is also a yes-instance. 

{ For the backward direction ($\Leftarrow$), it is sufficient to observe that $Z$ is disjoint from every obstruction of $G$.
A crucial observation is that $Z$ is disjoint from every cycle of $G$.
Hence, $Z$ is disjoint from every induced cycles of length at least 4 in $G$.
Consider a diamond $O$ in $G$.
Observe that every vertex of a diamond participates in a triangle in $O$ itself.
But, no vertex of $Z$ participates in any cycle.
Therefore, $Z$ is disjoint from every diamond of $G$.
Consider a paw $O$ of $G$.
All but one vertex of $O$ participates in a triangle in $O$ itself.
Since the vertices of $Z$ do not participate in any cycle, any $y \in O \cap Z$ is a pendant vertex of $O$.
But $N_G(Z) = \{v_2\}$.
But $v_2$ also does not participate in any cycle of $G$.
Hence, $v_2$ cannot participate in the triangle of $O$.
Therefore, $Z$ is disjoint from paw.
As $Z$ is disjoint from every obstruction of $G$, due to the second item of Observation \ref{observation:subgraph-yes-instance}, the backward direction holds true.}
%
\end{proof}

Our previous reduction rule has illustrated that we can shorten a long degree-2-tail to length at most two.
Now, we consider a degree-2-overbridge $P$ of length $\ell$.
Our next lemma gives us a structural characterization that if we contract all but a few vertices of $P$, then the set of all paws and diamonds remain the same even after deleting those vertices.  

\begin{lemma}
\label{lemma:rr7-paws-diamonds-same}
 Let $P = (v_1, v_2, \ldots, v_{\ell})$ be a degree-2-overbridge { with $\ell$ vertices in $G$ such that $\ell \geq 5$, $v_1 \neq v_{\ell}$,} and $Z = \{v_3, v_4, \dots, v_{\ell-2}\}$.
Consider $G'$ the graph obtained by deleting the vertices of $Z$ and then adding the edge $v_2 v_{\ell-1}$.
Then the following statements hold true.
\begin{enumerate}[(i)] 
	\item\label{paw-diamond-G-Z} Every paw and every diamond of $G$ is disjoint from $Z$.
	\item\label{paw-diamond-G-prime} { No paw or diamond in $G'$ contain the edge $v_2 v_{\ell-1}$.}
	\item\label{same-diamond-paw} The set of paws and diamonds in both $G$ and $G'$ are the same. 
	\item\label{same-double-edge} {The set of double-edges of $G$ and $G'$ are the same.}
\end{enumerate}
\end{lemma}


\begin{proof}
We prove the items in the given order.
\begin{enumerate}[(i)]
	\item Observe that every vertex of $Z$ has degree two in $G$ and $N_G(Z)$ has two vertices $v_2$ and $v_{\ell-1}$ both having degree two in $G$ as well. 
		Moreover, both in a paw or in a diamond, none of the vertices in the graph is such that it has a degree-2 vertex and both its neighbors also have degree at most two.
		Hence, $G$ cannot have an induced paw or an induced diamond that intersects $Z$.
		
		\item 
		{ For the sake of contradiction, suppose that there is a paw or a diamond of $G'$ that contains the edge $v_2 v_{\ell - 1}$.
		Consider a paw (or the diamond) of $G'$ that contains the edge $v_2 v_{\ell - 1}$.
		Since $v_2$ and $v_{\ell-1}$ are adjacent to only $v_1$ and $v_{\ell}$, respectively, in $G'$, the paw or diamond should contain one of $v_1$ or $v_{\ell}$.
		Now, we observe that in both paw and diamond graphs, if a subset of the graph induces a $P_3$, then the middle vertex of the $P_3$ must have degree-3.
		Considering the vertices $v_1, v_2$ and $v_{\ell - 1}$ that induces a $P_3$, but $v_2$ has degree two in $G$.
		Hence there cannot exist a paw or a diamond that contains $v_1, v_2, v_{\ell-1}$.
		By symmetry, considering $v_2, v_{\ell - 1}$ and $v_{\ell}$ that induces a $P_3$, note that the vertex $v_{\ell-1}$ has degree two in $G$.
		Hence, there cannot exist a diamond or a paw in $G'$ containing the vertices $v_2, v_{\ell-1}$ and $v_{\ell}$.
		Thus, we conclude that there are no paws or diamonds in $G'$ containing the edge $v_2 v_{\ell-1}$}.

		\item 
		Consider an arbitrary paw or an arbitrary diamond $O$ of $G$.
		Due to (\ref{paw-diamond-G-Z}) of this lemma, $O$ cannot intersect $Z$.
		Hence, the edge $v_2v_3$ cannot be present in $O$.
		Similarly, the edge $v_{\ell-2}v_{\ell-1}$ cannot be present in $O$.
		Moreover, observe that any paw $O$ of $G$ cannot contain both $v_2$ and $v_{\ell-1}$ together.
		{ Additionally, it follows from (\ref{paw-diamond-G-prime}) of this lemma that no paw or diamond $O'$ of $G'$ can contain the edge ${v_2 v_{\ell-1}}$.
		Hence, no paw or diamond of $G'$ contain both $v_2$ and $v_{\ell - 1}$ together.}
		If $v_2, v_{\ell-1} \in O$, then $O$ must contain $v_1$ as well as $v_{\ell}$.
		In such a case, $O = \{v_1,v_2,v_{\ell-1}, v_{\ell}\}$ that neither induces a paw nor induces a diamond both in $G$ and in $G'$.
		Therefore, any paw (and diamond) of $G$ is a paw (and diamond) of $G'$ and any paw/diamond of $G'$ is a paw/diamond of $G$ respectively.
		
		\item For this item, consider any double-edge of $G$ and $G'$.
		As $P = (v_1,\ldots,v_{\ell})$ is a degree-2-over bridge. 
		As the vertices $v_2,\ldots,v_{\ell-1}$ are of degree exactly two, there is no double-edge of $G$ that intersects $Z \cup \{v_2, v_{\ell-1}\}$.
		Hence, no double-edge of $G'$ intersects $v_2, v_{\ell-1}$.
		Therefore, the set of double-edges of $G$ and $G'$ are the same.
\end{enumerate}
This completes the proof of the lemma.
\end{proof}

Our next reduction rule exploits the above lemma and reduces the length of a degree-2-overbridge to at most four.

\begin{reduction rule}
\label{rr7}
Let $P = (v_1, v_2, \ldots, v_{\ell})$ be a degree-2-overbridge { with $\ell \geq 5$ vertices in $G$} and $Z = \{v_3, v_4, \dots, v_{\ell-2}\}$. 
Suppose that $G'$ is the graph obtained by removing $Z$ and adding the edge $v_2 v_{\ell-1}$.
The new instance is $(G', k)$.  
We refer to Figure \ref{fig:rr7-long-path} for an illustration.
\end{reduction rule}

\begin{figure}[t]
\centering
    \includegraphics[scale=0.3]{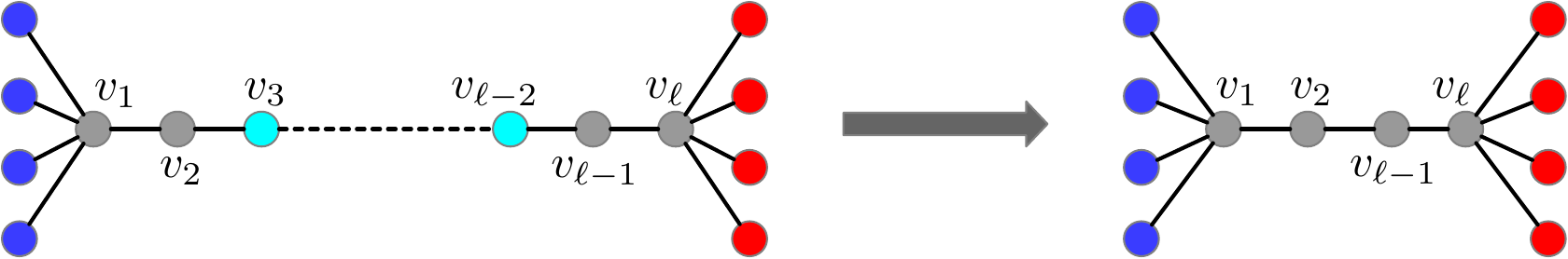}
    \caption{An illustration of applying Reduction Rule \ref{rr7}.}
    \label{fig:rr7-long-path}
\end{figure}

\begin{lemma}
Reduction Rule \ref{rr7} is safe.
\end{lemma}

\begin{proof}
Since the degree-2-overbridge $P$ of length $\ell$ has been converted  to a degree-2-overbridge of length four in $G'$, let $P' = (v_1, v_2, v_{\ell-1}, v_{\ell})$ be the degree-2-overbridge of length four.

First, we give the forward direction ($\Rightarrow$) of the proof.
Let $X$ be a solution to $(G, k)$.
We define $X'$ as follows. 
\begin{itemize}
	\item If $v_2 \in X$, set $X'$ to $(X \setminus \{v_2\}) \cup \{v_1\}$.
	\item  If $v_{\ell-1} \in X$, set $X'$ to $(X \setminus \{v_{\ell-1}\}) \cup \{v_\ell\}$.
	\item If $v_2, v_{\ell-1} \in X$, set $X'$ to $(X \setminus \{v_2, v_{\ell-1}\}) \cup \{v_1, v_2\}$.
	\item  Else, if $X \cap Z \neq \emptyset$, set $X'$ to $(X \setminus Z) \cup \{v_1\}$.
	\item Else,  set $X'$ to $X$.
\end{itemize}

Clearly, $|X'| \leq |X|$. 
Observe that $v_2, v_{\ell-1} \notin X'$.
We now claim that $X'$ is a solution to the instance $(G', k)$. 
Targeting a contradiction, suppose that there is a connected component of $G'-X'$ that is neither a clique nor a tree.
{%
As $X$ is a solution to $(G, k)$, $X$ intersects all double-edges of $G$.
Due to item (\ref{same-double-edge}) of Lemma \ref{lemma:rr7-paws-diamonds-same}, the set of all double-edges of $G$ and $G'$ are the same.
Hence, $X$ intersects all double-edges of $G'$.
Also, $X'$ contains all vertices of $X$ except the vertices of $Z \cup \{v_2, v_{\ell-1}\}$.
Therefore, $X'$ intersects all double-edges of $G'$.
So, we assume that $G' - X'$ is a simple graph.}

As we assume $X'$ is not a solution to $G'$ but $G' - X'$ is s simple graph, it follows from Lemma \ref{ctvd-forbidden-char} that $G' - X'$ contains a paw, or a diamond or a cycle $C_i$ with $i \geq 4$ as an induced subgraph.
It follows from item (\ref{paw-diamond-G-prime}) of Lemma \ref{lemma:rr7-paws-diamonds-same} that a diamond or a paw from $G'$ cannot contain both $v_2$ as well as $v_{\ell-1}$.
If a paw $Q$ of $G' - X'$ contains $v_2$ or $v_{\ell-1}$, then $Q$ must contain $v_1$ or $v_{\ell}$, respectively.
We consider both these cases one by one.

Let $Q$ be a paw of $G' - X'$ and $v_1, v_2 \in Q$.
Then $v_1, v_2 \notin X'$ and $Q$ is also a paw of $G'$.
It follows from item (\ref{same-diamond-paw}) of Lemma \ref{lemma:rr7-paws-diamonds-same}  that $Q$ is a paw of $G$ as well.
Since we have assumed that $v_1, v_2 \notin X'$, it must be that $v_1, v_2 \notin X$ (otherwise, $v_2 \in X$ implies $v_1 \in X'$ and $v_1 \in X$ implies $v_1 \in X'$).
As the other vertices of $X'$ are the same as $X$ except the vertices of $Z \cup \{v_2, v_{\ell-1}\}$, and $v_1, v_2 \in Q$, and $v_1, v_2 \notin X$, it ensures us that the vertex set $Q$ is also present in $G - X$.
Then, $G - X$ contains a paw contradicting our initial assumption that $X$ is a (clique, tree)-deletion set of $G$.
Similarly, in case a paw $Q$ of $G' - X'$ contains $v_{\ell-1}$ and $v_{\ell}$, then we can use similar arguments to justify that $Q$ is a paw in $G - X$ arriving at a contradiction that $X$ is a (clique, tree)-deletion set of $G$.

Observe that there cannot exist a diamond $O$ of $G' - X'$ such that $v_1, v_2 \in O$ or $v_{\ell-1}, v_{\ell} \in O$.
Hence, the only possibility is that  $G'-X'$ contains an induced cycle $C_i$ with $i \geq 4$.
The cycles in $G'-X'$ that are not present in $G-X$ must contain $v_2$ or $v_{\ell-1}$ and therefore, the entire degree-2-overbridge $P'$.
Then, any such chordless cycle of $G' - X'$ must also contain $v_1, v_{\ell}$, implying that $v_1, v_{\ell} \notin X'$.
In such a case, by construction of $X'$ from $X$, $Z \cup \{v_2, v_{{\ell}-1}\} \cap X = \emptyset$. 
In particular, $v_1, v_{\ell} \notin X$, (as otherwise by construction $v_1 \in X'$ or $v_{\ell} \in X'$).
Hence, $P \cap X = \emptyset$.
Let this cycle in $G'-X'$ be $C'$ that contains $P'$.
We obtain $C$ from $C'$ by replacing the degree-2-overbridge $P'$ by the degree-2-overbridge $P$. 
From the definition of $X'$, it follows that $C$ is an induced cycle of length at least four present in $G-X$.
It contradicts that $X$ is a (clique, tree)-deletion set to $(G, k)$.
Thus, $X'$ is a solution to $(G', k)$.

Conversely, for the backward direction ($\Leftarrow$), let $X'$ be a (clique, tree)-deletion set of size at most $k$ in $G'$. 
We claim that $X'$ is also a solution of $G$. 
Targeting a contradiction, suppose that $X'$ is not a solution to $G$.
Observe from item (\ref{same-double-edge}) of Lemma \ref{lemma:rr7-paws-diamonds-same} that every double-edge of $G'$ is also a double-edge of $G$.
As $X'$ intersects all double-edges of $G'$, $X'$ intersects all double-edges of $G$.
So, $G - X'$ is a simple graph. 
Then by Lemma~\ref{ctvd-forbidden-char}, $G - X'$ contains a paw, or a diamond or a cycle $C_i$ with $i \geq 4$ as an induced subgraph.
Such an obstruction (i.e. paw or diamond or chordless cycle) also exists in $G$.
 From item (\ref{same-diamond-paw}) of Lemma \ref{lemma:rr7-paws-diamonds-same}, the sets of paws and diamonds in $G$ and $G'$ are the same.
As $X'$ intersects all diamonds and all paws of $G'$, $X'$ intersects all paws and all diamonds of $G$ as well.
Hence, the only possible obstruction for $G - X'$ is an induced cycle of length at least four.
Note that the set of induced cycles in $G$ that do not contain any vertex from $P \setminus \{v_1, v_{\ell}\}$ are also present in $G'$.
As $X'$ is a solution to $G'$, $X'$ intersects all induced cycles of $G'$.
Hence, the the chordless cycle $C$ of in $G - X'$ must contain some vertex from $Z \cup \{v_2, v_{{\ell} - 1}\}$.
Then, $C$ is an induced cycle of length at least $4$ that contains the entire degree-2-overbridge $P$ and $X'$ is disjoint from $P$.
 Then the cycle $C'$ obtained by replacing $P$ with $P'$ is present in $G'-X'$.
 Moreover, $C'$ is also an induced cycle of $G' - X'$ containing at least four vertices contradicting that $X'$ is a solution to $(G', k)$.
This completes the proof.
\end{proof}

We use another reduction rule whose safeness is similar to the crown reduction rule for {\VC}.

\begin{reduction rule}
\label{rule:pendant-cycle}
If $D$ is a pendant hole attached to $v$, then delete $D$ from the graph.
The new instance is $(G - D, k - 1)$.
\end{reduction rule}

\begin{lemma}
\label{lemma:safeness-pendant-cycle}
Reduction Rule \ref{rule:pendant-cycle} is safe.
\end{lemma}

\begin{proof}
The backward direction ($\Leftarrow$) is trivial.
For the forward direction ($\Rightarrow$), observe that $D$ is an induced cycle of $G$ with at least 4 vertices.
Hence, at least one vertex from $D$ must be part of the solution.
Since $D$ is attached to $v$, deleting $v$ creates a connected component $D \setminus \{v\}$ that is a path that has no vertex adjacent to any vertex other than $v$ in the graph.
Therefore, if some vertex $w$ other than $v$ is in a solution $X$ of size at most $k$, then $(X \setminus \{w\}) \cup \{v\}$ is a solution to $(G, k)$.
Hence, this reduction rule is safe.
\end{proof}

%


\subsection{Bounding the Clique Vertices in $G - S$}
\label{sec:clique-vertex-bounding}

Let $V_1 \subseteq V(G) \setminus S$ denote the set of vertices of the connected components that form cliques of size at least 3.
We now bound the number of connected components in $G[V_1]$ (which are cliques).
Let us create an auxiliary bipartite graph $H = (S,\CC)$ with $S$ on one side and $\CC$ having a vertex set corresponding to each of the clique connected components in $V_1$
  on the other side. 
  We add an edge $(s,C)$ with $s \in S$ and $C \in \CC$ if $s$ is adjacent to at least one vertex in $C$.
We now show how to ensure that $|\CC| \leq 2|S|$.
Note that by Reduction Rule \ref{rr1}, no component in $\CC$ is an isolated vertex in $H$. So, we have the following reduction rule, where we rely on the Expansion Lemma. 

\begin{reduction rule}
\label{rr2}
If $|\CC| \geq 2|S|$, then call the algorithm provided by the
$q$-Expansion Lemma with $q = 2$ (Lemma~\ref{lemma:expansion}) to compute sets $X \subseteq S$ and $\mathcal{Y} \subseteq \CC$ such that there is a $2$-expansion $M$ of $X$ into $\mathcal{Y}$ in $H$ and $N_H(\mathcal{Y}) \subseteq X$.
The new instance is $(G - X, k - |X|)$.
\end{reduction rule}



\begin{lemma}
\label{lemma:clique-expansion-rule-safeness}
Reduction Rule \ref{rr2} is safe.
\end{lemma}

\begin{proof}
In one direction, it is clear that if $Z$ is a solution to $(G - X, k - |X|)$, then $Z \cup X$ is a solution to $(G, k)$ as $|Z \cup X| \leq (k-|X|) + |X| = k$.
For the other direction, let $Z$ be a solution to $(G, k)$.
Let $Y$ denote the set of vertices in the connected components of $\mathcal{Y}$ in $H$.
We denote $Z' = (Z \setminus Y) \cup X$.
Since there is a $2$-expansion $M$ of  $X$ into $\mathcal{Y}$ in $H$, it follows from Lemma \ref{lemma:expansion} that for every $x \in X$,  the there are two different connected components $C_1$ and $C_2$ of $\mathcal{Y}$ such that $x$ is adjacent to $C_1$ and $C_2$ in $H$.
 Thus, $x$ is adjacent to $u \in C_1$ in $G$, and is adjacent to $v \in C_2$ in $G$.

{ For every $x \in X$, one of the three conditions must hold.
\begin{description}
  	\item[(i)]  
  	There is a paw that contains $x$ and the vertices from two connected components $C_x^1, C_x^2 \in \mathcal{Y}$ such that $(x, C_x^1), (x, C_x^2) \in M$.
  	
  	\item[(ii)] There is a paw containing $x$ and 3 vertices of a connected component $C_x \in \mathcal{Y}$ such that $(x, C_x) \in M$.

  	\item[(iii)] There is a double-edge containing $x$ and a connected component $C_x \in \YY$ such that $(x, C_x) \in M$.
  	
  	To see this, if $xu$ or $xv$ is a double-edge such that $u \in C_1$ and $v \in C_2$, then there is a component in $\mathcal{Y}$ containing a vertex 
\end{description} 

We explain why one of the above three conditions must hold.
Suppose that $x$ is adjacent to (at least) two vertices $u, u'$ of $C_1$ and $xu, xu', xv$ have multiplicity one each, then $u',u,x,v$ form a paw.
  	Hence, there is a paw that contains $x$ and vertices of two connected components in $\mathcal{Y}$.
Hence, (i) holds true.

Suppose that $x$ is adjacent to $u$ alone in $C_1$ and $xu$ has multiplicity one.
 	Since $C_1$ is a clique of size at least $3$, there exist vertices $u_1,u_2$ such that $u_1, u_2, u$ form a triangle.
 	Thus, $u_1, u_2, u, x$ form a paw.
 	Hence, there is a paw containing $x$ and 3 vertices from a connected component of $\mathcal{Y}$.
 	Hence, (ii) is satisfied.
  
Finally, if $xu$ or $xv$ is a double-edge, then there is a double-edge that contains $x$ and a vertex $u \in C_1$.
Hence, condition (iii) is satisfied.}

{
As the above 3 cases are mutually exhaustive,
 therefore, for every $x \in X$, one of the following three things are satisfied}.
 
Furthermore, all these paws, and double-edges are pairwise vertex-disjoint.
As $Z$ is a solution to $(G, k)$, it must hit all such paws { and double-edges}.
Thus, if a vertex $x \notin Z$, then $Z$ must contain 
\begin{itemize}
	\item at least one vertex from a paw containing $x$ and vertices from $C_x^1 \cup C_x^2$, or
	\item at least one vertex from a paw containing $\{x\} \cup C_x$ such that $(x, C_x) \in M$, or
	\item at least one vertex from a double-edge $xu$ such that $u \in C_x$ with $(x, C_x) \in M$.
\end{itemize}  
Hence, $|X \setminus Z | \leq |Z \cap Y|$ and thereby, $|Z'| = |Z| - |Z \cap Y| + |X \setminus Z| \leq |Z | \leq k$.
We claim that $G - Z'$ is a collection of cliques and trees.
Targeting a contradiction, suppose not.
Then $G-Z'$ contains an obstruction for {\ctvd} containing a vertex in $Y$.
Let $O$ be the vertex set corresponding to such an obstruction.
Note that $O$ cannot be in $G - S$ as it is a collection of cliques and trees.
Thus $O$ contains a vertex in $S$.
But since $N_H(\mathcal{Y}) \subseteq X$, hence $N_G(Y) \subseteq X$.
As all the obstructions of {\ctvd} are connected graphs, an obstruction $O$ must contain vertices from in $S \setminus (X \cup Y)$.
In particular, $O$ must contain vertices from $S \setminus (Z \cup Y)$.
Hence, $O$ is present in $G - (Z' \cup Y)$.
As $O$ contains vertices from $Z' \setminus (X \cup Y)$, and $Z'$ contains all vertices of $Z$ except the vertices of $Y$, $O$ is  also present in $G - Z$.
This contradicts that $Z$ is a solution to $G$.
Thus, we conclude that $Z'$ is a solution to $(G, k)$, and as $X \subseteq Z'$, we have that $Z' \setminus X$ is a solution to $(G - X, k - |X|)$.
Thus, $(G - X, k - |X|)$ is a yes-instance.
\end{proof}

Thus, we have the following observation.

\begin{observation}
\label{lemma:clique-number-atleast-3}
After exhaustive applications of Reduction Rules \ref{rr1}- \ref{rr2}, $|\CC| \leq 8k$.
\end{observation}

We now give one of the most crucial reduction rules that gives us an upper bound the size of every clique in $G[V_1]$.
First, we perform the following marking scheme for each of the cliques in $G[V_1]$.
{ For every connected component $K$ of $G[V_1]$, we use the following procedure {\sf Mark-Clique($K$)} that works as follows.}



\paragraph{{Procedure {\sf Mark-Clique($K$)}}:}
{
For every non-empty subset $Z \in {{S}\choose{\leq 3}}$, for every function $f: Z \rightarrow \{0,1\}$, let $K_{f, Z}$ be those vertices $v \in K$ that satisfies the following properties:
\begin{itemize}
\item if $f(z) =1$, then $v$ is adjacent to $z$ such that $vz$ is not a double-edge. 
\item if $f(z) =0$, then $v$ is not adjacent to $z$.
\end{itemize}}

{ In the first step, for every $v \in S$, we mark every vertex $u \in K$ such that there is a double-edge between $u$ and $v$.
Subsequently, for every nonempty set $Z \in {{S}\choose{\leq 3}}$ and for every $f: Z \rightarrow \{0, 1\}$, we arbitrarily mark $\min\{|K_{f, Z}|, k+4\}$ vertices from $K_{f, Z}$.
This completes the description of marking procedure.}

\medskip

As Reduction Rule \ref{rule:high-degree-double-edge} is not applicable, for every $v \in S$, there are at most $k$ vertices from $K$ that are adjacent to $v$ with a double-edge.
Hence, in the first part, we have marked at most $4k^2$ vertices from $K$.
Subsequently, we have marked at most ${\eps}^*(k) = (2^3 {4k \choose 3} + 2^2 {4k \choose 2} + 2{{4k}\choose{1}})(k+4)$ vertices in $K$.
We set ${\eps}(k) = {\eps}^*(k) + 4k^2 = (2^3 {4k \choose 3} + 2^2 {4k \choose 2} + 2{{4k}\choose{1}})(k+4) + 4k^2$ which is $\OO(k^3)$.
Let $v \in K$ be a vertex that is not marked by the above procedure {\sf Mark-Clique($K$)}.
The following set of lemmas illustrates that $v$ is an irrelevant vertex of $G$.
\begin{figure}[t]
    \centering
    \includegraphics[scale=0.35]{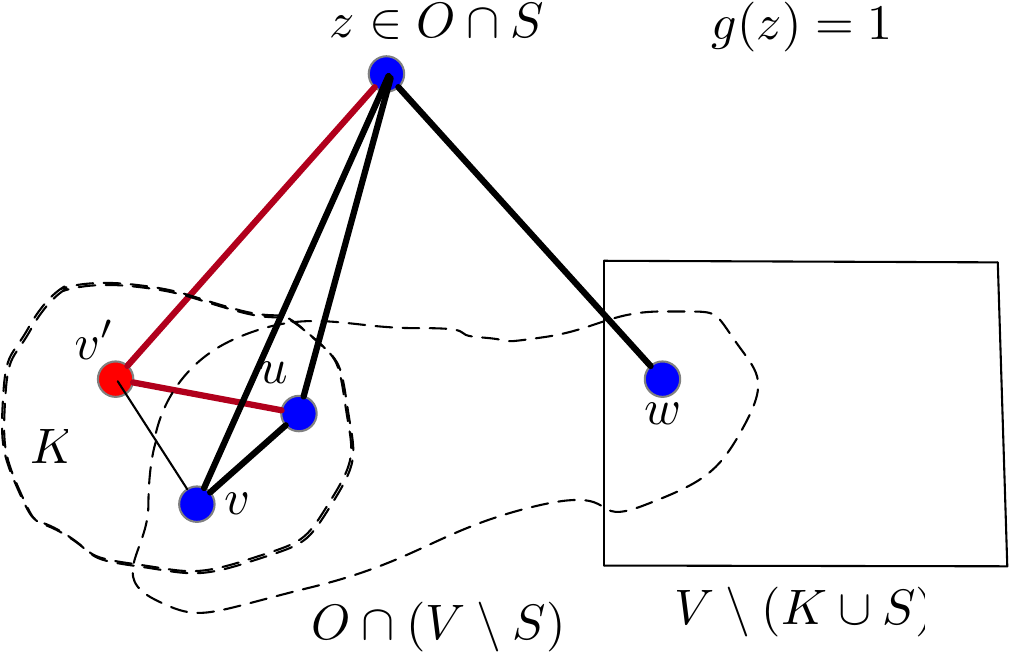}
    \caption{Illustration for obstruction $O = \{z, u, v, w\}$ where we find an isomorphic obstruction $O = \{z, u, v', w\}$. Here, $v'$ marked in the procedure {\sf Mark-Clique($K$)} with respect to $O_S = O \cap S$ and $g$ and the unmarked vertex $v$.}
    \label{fig:clique-tree-unmarked-vertex_new}
\end{figure}

\begin{lemma}
\label{lemma:unmarked-fact-1clique-marking-oneshot}
Let $S$ be a (clique, tree)-deletion set of at most $4k$ vertices and $K$ be a connected component of $G - S$ that is a clique.
Moreover, let $v \in K$ be a vertex that is not marked by the procedure {\sf Mark-Clique($K$)} and $X \subseteq V(G) \setminus \{v\}$ be a set of at most $k$ vertices.
If $G - X$ is a simple graph that has a vertex subset $O$ and $G[O]$ is isomorphic to a $C_4$, or a diamond or a paw, then $G - (X \cup \{v\})$ also contains a $C_4$, or a diamond, or a paw as an induced subgraph.
\end{lemma}

\begin{proof}
Targeting a contradiction, suppose that $G - X$ is a simple graph that has a $C_4$, or a diamond, or a paw the vertex set of which is $O$ but $G - (X \cup \{v\})$ has none of paw, diamond, or $C_4$ as induced subgraphs.
Note  that $v \in O$ (hence $v \in O \setminus S$) as otherwise $X$ is not a (clique, tree)-deletion set for $G - \{v\}$.

Let $R = V(G) \setminus (S \cup K)$. Also, let $O_S = O \cap S, O_K = O \cap K$ and $O_R = O \cap R$. We have, $O = O_S \uplus O_K \uplus O_R$.
{ Observe that $O$ cannot be contained in $G - S$.
Hence, $O$ must contain a vertex from $S$ and $O$ must contain a vertex from $K$, implying that $O_S, O_K \neq \emptyset$.
Hence, $|O_S| \leq 3$.}
Let us define a function $g : O_S \rightarrow \{0,1\}$ where for each $z \in O_S$, $g(z)=1$ if $z$ is adjacent to $v$ and $g(z)=0$ if $z$ is not adjacent to $v$. 
In addition, $v \in O_K$ and hence, $v \in K$.
In particular, $v$ is not a marked vertex from $K$ with respect to $g$ and $O_S$.
Therefore, $|K_{g, O_S}| > k + 4$ and the procedure {\sf Mark-Clique($K$)} has marked $k+4$ vertices in $K$ corresponding to the set $O_S$ and the function $g: O_S \rightarrow \{0,1\}$.

Since $|X| \leq k$ and $|O \setminus O_S| \leq 3$, there is at least one vertex $v' \in K$ with $v' \notin O$ that is marked corresponding to $O_S$ and $g$.
Let $O' = (O \setminus \{v\}) \cup \{v'\}$. Also, let $H = G[O]$ and $H' = G[O']$ (see Figure \ref{fig:clique-tree-unmarked-vertex_new} for an example).
We have the following claim.

\noindent
{\bf Claim:}
$H'$ is isomorphic to $H$.

\noindent
{\bf Proof of Claim:} We define a mapping $h: O \rightarrow O'$ where $h(u) = u$ if $u \in O \setminus \{v\}$ and $h(v) = v'$. We now claim that $h$ is an isomorphism. For this, we need to prove that $u_1u_2 \in E(H)$ if and only if $h(u_1)h(u_2) \in E(H')$ for all $u_1, u_2 \in O$.
We have the following cases.

\begin{description}
	\item[Case 1:] - $u_1u_2 \in E(H)$ where $u_1, u_2 \in O \setminus \{v\}$.
Since, $u_1, u_2 \in O \setminus \{v\}$, $h(u_1) = u_1$ and $h(u_2) = u_2$. Thus $h(u_1)h(u_2) \in E(H')$.

	\item[Case 2:] - $u_1 u_2 \notin E(H)$ where $u_1, u_2 \in O \setminus \{v\}$.
	Same as before, since $h(u_1) = u_1$ and $h(u_2) = u_2$, it follows that $h(u_1)h(u_2) \notin E(H')$.
	
	\item[Case 3:] - $uv \in E(H)$ where $u \in O_S$.
	We have $h(u) = u$ and $h(v) = v'$. Since $g(u)=1$ (as $u$ is adjacent to $v$), and $v'$ is marked with respect to $O_S$ and $g$, $u$ is also adjacent to $v'$.
	Hence, $h(u) h(v) \in E(H')$.
	
	\item[Case 4:] - $u v \notin E(H)$ where $u \in O_S$.
	We have $h(u) = u$ and $h(v) = v'$. Since $g(u)=0$ (as $u$ is not adjacent to $v$), and $v'$ is marked with respect to $O_S$ and $g$, $u$ is also not adjacent to $v'$. Thus $h(u) h(v) \notin E(H')$.
	
	\item[Case 5:] - $uv \in E(H)$ where $u \in O_K$.
	We have $h(u) = u$ and $h(v) = v'$. Since $u, v, v' \in K$, $u$ is adjacent to $v'$ as well.	Thus $h(u)h(v) \in E(H')$.
	
	\item[Case 6:] - $u v \notin E(H)$ where $u \in O_K$.
	Since $u, v \in K$, $u$ is adjacent to $v$. Thus, this case does not exist.
	
	\item[Case 7:] - $u v \in E(H)$ { where $u \in O_R$}.
	Since $u \in R$ and $v \in K$, $u$ is not adjacent to $v$ as $u$ is in a different component of $G-S$ than $K$. Thus, this case cannot occur.
	
	\item[Case 8:] - $u v \notin E(H)$ where $u \in O_R$.
	We have $h(u) = u$ and $h(v) = v'$. Since $v, v' \in K$ and $u$ is a component of $G-S$ that is different from $K$, $u$ is not adjacent to $v'$ as well.	Thus $h(u) h(v) \notin E(H')$.
\end{description}

We have exhaustively gone over the cases of all possible pair of vertices in $H$. Thus, we have proved the claim that $H'$ is isomorphic to $H$.

Thus, the obstruction $O'$ is isomorphic to $O$ and $O'$ is present in the graph $G - (X \cup \{v\})$. This is a contradiction to the initial assumption that $G - (X \cup \{v\})$ has no paw, no diamond, no $C_4$ as induced subgraphs.
\end{proof}

Our previous lemma has illustrated that if $G - X$ has a paw or a diamond or a $C_4$ as an induced subgraph, then $G - (X \cup \{v\})$ also has a paw or a diamond or a $C_4$ respectively.
We will now illustrate and prove an analogous statement when $G - X$ has an induced cycle of length larger than 4. 
We begin with the following observation.

\begin{observation}\label{observation:cycle-5-above-hole-diamond}
Let $C = v - u-u_1-u_2- \dotsc - u'-v$ be a cycle of length at least 5 in $G$ where the path $P = u-u_1-u_2- \dotsc - u'$ is an induced path in $G$.
Then there exists a cycle of length at least $4$ or a diamond as an  induced subgraph in $G$.
\end{observation}

\begin{figure}[t]
\centering
	\includegraphics[scale=0.3]{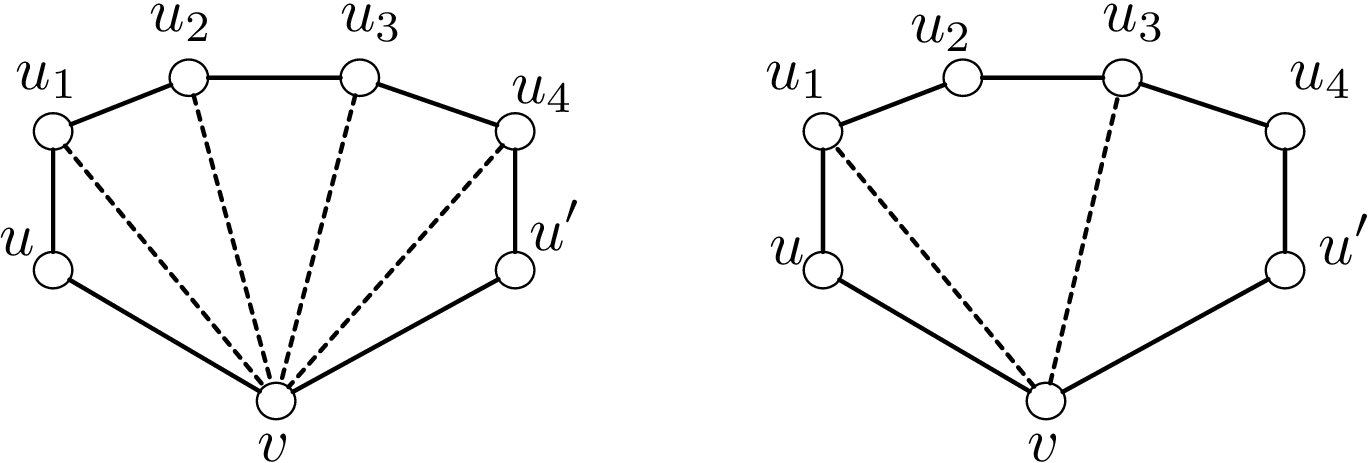}
	\caption{An illustration of Observation \ref{observation:cycle-5-above-hole-diamond} proof. Figure on the left side illustrates Case-(i) of the proof. Figure on the right side illustrates an example for Case-(ii) of the proof. If we choose $w = u_2$, then $p_w = u_1$ and $s_w = u_3$. Similarly, if $w = u_4$, then $p_w = u_3$ and $s_w = u'$.}
\label{fig:long-cycle-clique-part}
\end{figure}

\begin{proof}
Note that the only edges of $C$ in $G$ other than the edges of the cycle $C$ itself are those from $v$ to the other vertices in $C$.
In other words, if there is an edge between two vertices of $C$ that is not an edge of the cycle $C$ itself is a chord in $C$.
As $C$ has length at least 5, $u_1$ and $u_2$ exist in $C$ such that $u_1 \neq u, u'$ and $u_2 \neq u, u'$.
We have two cases.
\begin{description}
	\item[Case-(i):] This case occurs when $v$ is adjacent to all the vertices of $C$.
	Then notice that the graph induced by the vertices $v$ and the first three vertices $u, u_1, u_2$ is a diamond because $u - u_1 - u_2$ is an induced path.
	Therefore, a diamond exists in $G$ as an induced subgraph.
	We refer to the left part of Figure \ref{fig:long-cycle-clique-part} for an illustration of this case.

	\item[Case-(ii):] This case occurs when there is $w \in C$ such that $v$ is not adjacent to $w$.
	Note that $w \neq u$ and $w \neq u'$.
	 Traveling from $w$ in the left direction via $P$, let $p_w$ be the first vertex that is adjacent to $v$. Such a vertex must exist as $v$ is adjacent to $u$. 
	 Similarly, traveling from $w$ in the right direction via $P$, let $s_w$ be the first vertex that is adjacent to $v$. Such a vertex must exist as $v$ is adjacent to $u'$. 
	 Note that $p_w \neq s_w$. Let us look at the cycle $C' = v-p_w- \dotsc -w- \dotsc -s_w-v$. 
	 Since $P$ is an induced path in $G$ and $v$ is not adjacent to any  vertex in $C'$ other than $p_w$ and $s_w$, $C'$ is an induced cycle of $G$.
	 Since the vertices $v, p_w, w$ and $s_w$ are all distinct,  $C'$ of length at least $4$.
	 We refer to the right part of Figure \ref{fig:long-cycle-clique-part} for an illustration of this case.
\end{description}
As the above cases are mutually exhaustive, this concludes the proof.
\end{proof}

Using the above observation, we can prove the following lemma.

\begin{lemma}
\label{lemma:unmarked-fact-1clique-marking-2}
Let $S$ be a (clique, tree)-deletion set of at most $4k$ vertices and $K$ be a connected component of $G - S$ that is a clique.
Moreover, let $v \in K$ be a vertex that is not marked by the procedure {\sf Mark-Clique($K$)} and $X \subseteq V(G) \setminus \{v\}$ be a set of at most $k$ vertices.
If $G - X$ has an induced cycle of length at least 5, then there exists a cycle of length at least 4 or a diamond as an induced subgraph in $G - (X \cup \{v\})$.
\end{lemma}

\begin{proof}
{
Suppose that the premise of the statement is true but for the sake of contradiction, we assume that $G - (X \cup \{v\})$ does not have  cycles of length at least 4 or diamonds as induced subgraphs}.
Since $v$ is the only vertex that is in $G$ but not in $G - \{v\}$, it follows that there is an induced cycle $C$ of length at least 5 in $G - X$ such that $v \in C$.
Note that $C$ has at most two vertices from $K$ including $v$ as $K$ is a clique.
Furthermore, it must have two non-adjacent vertices from $S$ as otherwise $C$ contains a triangle or an induced $C_4$, contradicting that $C$ is an induced cycle of length at least 5. 
There are two cases. 
\begin{description}
	\item[Case (i):] The first case is $|C \cap K| = 1$ and let $C \cap K = \{v\}$.
	{ Then, $v$ has two neighbors that are in $S \cap C$.
	Consider $z_1, z_2 \in S \cap C$ are adjacent to $v \in C$.}
	Let us define a function $f: \{z_1, z_2\} \rightarrow \{0,1\}$ with $f(z_1) = 1$ and $f(z_2) = 1$.
	{
	For the set  $Z = \{z_1, z_2\}$ and the function $f$, the vertex $v \in K_{f, Z}$ but $v$ is unmarked by the procedure {\sf Mark-Clique($K$)}.
	 Then, $|K_{f, Z}| > k+4$.
	Hence, there are $k+4$ vertices in $K_{f, Z}$ that are marked by the procedure {\sf Mark-Clique($K$)}.
	All the marked vertices are in $K \setminus \{v\}$ out of which at most $k$ vertices are in $X$.
	Hence, there is $v' \in K_{f, Z} \setminus X$ such that $v'$ is adjacent to both $z_1$ and $z_2$ and is marked by the procedure.}
	Let us look at the cycle $C' = (C \setminus \{v\}) \cup \{v'\}$.
	Observe that $C'$ is a cycle in $G - (X \cup \{v\})$ as a subgraph.
	Note that $C'$ is a cycle where the path in $C'$ from $z_1$ to $z_2$ that excludes $v'$ is an induced path in $G-(X \cup \{v\})$. 
	There could be edges from $v'$ to other vertices of $C'$. 
	By Observation \ref{observation:cycle-5-above-hole-diamond}, there exists a vertex subset $J'$ in $G - (X \cup \{v\})$ that induces an induced cycle of length at least $4$ or a diamond.
	Since $J' \cap (X \cup \{v\}) = \emptyset$ and induces a cycle of length at least 4 or a diamond, this contradicts our initial assumption that $G - (X \cup \{v\})$ does not have  cycles of length at least 4 and diamonds as induced subgraphs.
	
	\item[Case (ii):] The second and last case is $|C \cap K| = 2$.
	Let $v, x \in C \cap K$.
	Observe that $v$ and $x$ are two consecutive vertices in $C$.
	Since $vx \in E(G)$, it must be that $v$ is adjacent to $z_1 \in S \cap C$ and $x$ is adjacent to $z_2 \in S \cap C$.
	As $C$ is an induced cycle of length at least 5, it must be that $z_1$ is not adjacent to $x$ and $z_2$ is not adjacent to $v$.
	
	{
	Let us define a function $f: \{z_1, z_2\} \rightarrow \{0,1\}$ with $f(z_1) = 1$ and $f(z_2) = 0$.
	For the set $Z = \{z_1, z_2\}$ and the function $f$, the vertex $v \in K_{f, Z}$.
	But $v$ is unmarked by the procedure {\sf Mark-Clique($K$)}. 
	Hence, there are $k+4$ vertices in $K_{f, Z}$ that are marked by the procedure.
	All the marked vertices are in $K \setminus \{v\}$ out of which at most $k$ vertices are in $X$.
	Hence, there is $v' \in K_{f, Z} \setminus X$ such that $v'$ is not marked by the procedure}.
	
	We replace $v$ in $C$ by $v'$ to get a new cycle $C'$ that has the same number of vertices as $C$ (see Figure \ref{fig:clique-K-C5-case} for an illustration).
	Note that $C'$ is a cycle as a subgraph in $G - (X \cup \{v\})$ where the path from $z_1$ to $x$ that excludes $v'$ is an induced path in $G- (X \cup \{v\})$ and $|C'| = |C|$.
	There could be edges from $v'$ to other vertices of $C'$.
	By Observation \ref{observation:cycle-5-above-hole-diamond}, there exists a vertex subset $J'$ that induces a cycle of length at least $4$ or a diamond as  in $G - (X \cup \{v\})$.
	This contradicts our initial assumption that $G - (X \cup \{v\})$ does not have  cycles of length at least 4 and diamonds as induced subgraphs.
\end{description}
Since the above cases are mutually exhaustive, this completes the proof.
\end{proof}

\begin{figure}[t]
\centering
	\includegraphics[scale=0.3]{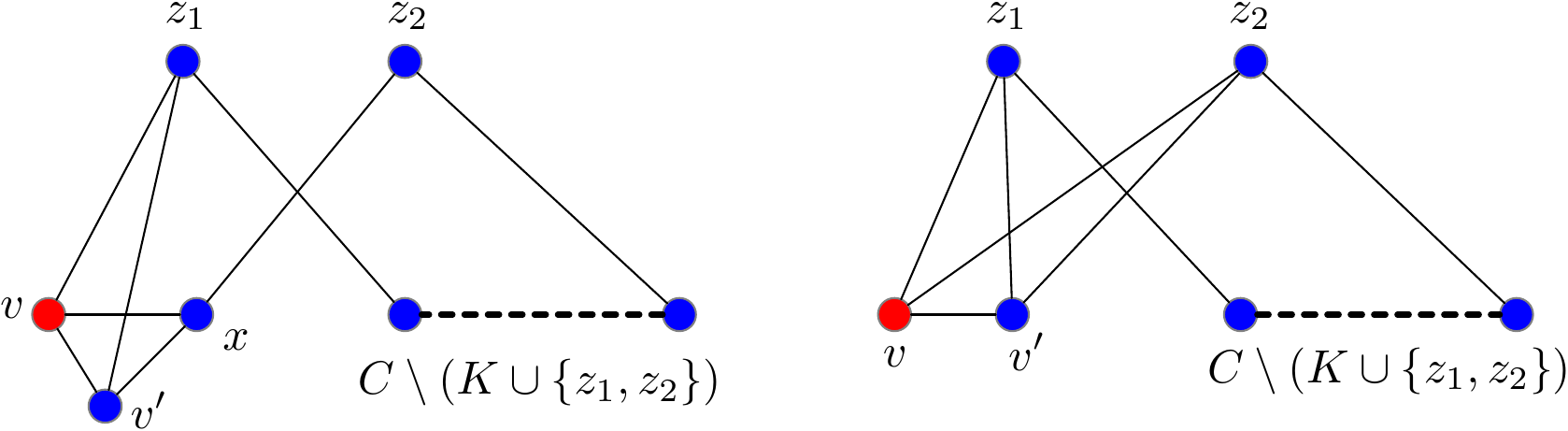}	
	\caption{An illustration of $C_5$.}
	\label{fig:clique-K-C5-case}
\end{figure}

Consider a connected component $K$ of $G - S$ such that $S$ is a (clique, tree)-deletion set of $G$ with at most $4k$ vertices and $K$ is a clique.
Lemma \ref{lemma:unmarked-fact-1clique-marking-oneshot} and Lemma \ref{lemma:unmarked-fact-1clique-marking-2} illustrate { that} if a vertex $v \in K$ is not marked by the procedure {\sf Mark-Clique($K$)}, then deleting $v$ from the graph is safe.
As a consequence of this, we have the following reduction rule the safeness of which follows from the above two lemmas.

\begin{reduction rule}
\label{rr6}
Let $K$ be a connected component in $G[V_1]$.
If $v \in K$ is an unmarked vertex after invoking the procedure {\sf Mark-Clique($K$)}, then remove $v$ from $G$.
The new instance is $(G-v,k)$.
\end{reduction rule}

\begin{lemma}
\label{lemma:rr6-safeness}
Reduction Rule \ref{rr6} is safe.
\end{lemma}

\begin{proof}
The forward direction ($\Rightarrow$) is trivial as if $X$ is a solution of size at most $k$ in $G$, then by Observation \ref{observation:subgraph-yes-instance}, $X \setminus \{v\}$ is a solution of size at most $k$ in $G-v$. 
For the backward direction ($\Leftarrow$), let $X$ be a solution of size at most $k$ in $G-v$. Targeting a contradiction, suppose $X$ is not a solution in $G$.
Due to our procedure {\sf Mark-Clique($K$)}, every $w \in K$ is marked when there is a double edge incident to $w$.
As $v$ is an unmarked vertex, there is no double edge incident to $v$ in $G$.
Therefore, $G - X$ is a simple graph.
{ As $X$ is not a solution and $G - X$ is a simple graph, there exists an obstruction $O$ of {\ctvd} in $G-X$ containing $v$ that is a diamond, or a paw, or $C_i$ where $i \geq 4$ due to Lemma \ref{ctvd-forbidden-char}}.
If $O$ is isomorphic to a $C_4$, or a diamond, or a paw, then by Lemma \ref{lemma:unmarked-fact-1clique-marking-oneshot}, $G-(X \cup \{v\})$ also contains a $C_4$, or a diamond, or a paw as induced subgraph.
It contradicts Lemma \ref{ctvd-forbidden-char} that $X$ is a solution to $G-\{v\}$.
Else, $O$ is isomorphic to $C_i$ where $i \geq 5$.
By Lemma \ref{lemma:unmarked-fact-1clique-marking-2}, $G-(X \cup \{v\})$ contains a $C_j$, where $j \geq 4$, or a diamond as induced subgraph.
This contradicts Lemma \ref{ctvd-forbidden-char} as $X$ is a solution to $G-\{v\}$.
Hence $X$ is a solution to $G$.
\end{proof}

We have the following lemma that bounds $|V_1|$, i.e. the number of vertices that are part of cliques of size at least 3 in $G - S$.

\begin{lemma}
\label{lemma:cliques-bound}
Let $G$ be the graph obtained after exhaustive application of Reduction Rules \ref{rr1} to \ref{rr6}.
Then $|V_1| \leq 8k(k+4)(8{{4k}\choose{3}} + 4{{4k}\choose{2}} + 2{{4k}\choose{1}}) + 32k^3$.
\end{lemma}

\begin{proof}
Recall that we defined $\eps(k) = 8{{4k}\choose{3}} + 4{{4k}\choose{2}} + 2{{4k}\choose{1}} + 4k^2$.
Since Reduction Rules \ref{rr1}-\ref{rr2} are not applicable, it follows from Observation \ref{lemma:clique-number-atleast-3} that the number of connected components in $G[V_1]$ is at most $8k$.
Every connected component of $G[V_1]$ is a clique.
For every connected component $K$ of $G[V_1]$, observe that the procedure {\sf Mark-Clique($K$)} marks at most $\eps(k)$ vertices from $K$.
Since Reduction Rule \ref{rr6} is not applicable, $G[V_1]$ has no unmarked vertices.
As $G[V_1]$ has at most $8k$ connected components, { the equation below follows.}

	$$|V_1| \leq 8k \eps(k) = 8k(k+4)\Bigg(8{{4k}\choose{3}} + 4{{4k}\choose{2}} + 2{{4k}\choose{1}}\Bigg) + 8k \cdot 4k^2$$
	
	$$= 8k(k+4)\Bigg(8{{4k}\choose{3}} + 4{{4k}\choose{2}} + 2{{4k}\choose{1}}\Bigg) + 32k^3.$$

This completes our claimed upper bound on the number of vertices in $V_1$.
\end{proof}

\subsection{Bounding the Tree Vertices in $G - S$}
\label{sec:tree-vertices-G-S}
In this section, we describe the set of reduction rules that we use to reduce the number of vertices that participate in the forests of $G - S$.
Let $V_2 = V  \setminus (S \cup V_1)$.  
Note that  $G[V_2]$ is the collection of trees in $G-S$.

\begin{reduction rule}\label{reduction-rule-leaf-neighbor-not-adj}
Let $v$ be a pendant vertex in a connected component $C$ of $G[V_2]$ such that neither $v$ nor its neighbor in $C$ is adjacent to any vertex of $S$. 
{
Then, delete $v$ from $G$ and the new instance is $(G - v, k)$.}
\end{reduction rule}

\begin{lemma}
Reduction Rule \ref{reduction-rule-leaf-neighbor-not-adj} is safe.
\end{lemma}

\begin{proof}
Let $u$ be the unique neighbor of $v$ in $G$.
For the forward direction ($\Rightarrow$), let $X$ be a solution of size $k$ in $G$. 
Since $G-v$ is an induced subgraph of $G$, it follows from Observation \ref{observation:subgraph-yes-instance} that $X \setminus \{v\}$ is a solution in $G-v$ as well.

Conversely, for the backward direction ($\Leftarrow$), let $X'$ be a solution of size $k$ in $G-v$. 
Targeting a contradiction, suppose that $X'$ is not a solution to $G$.
As $G - (X' \cup \{v\})$ is a simple graph, and $v$ is a pendant vertex of $G$, therefore, $G - X'$ is a simple graph as well.
As we assume that $X'$ is not a solution to $G$, there is a connected component of $G-X'$ is neither a clique nor a tree.
Then by Lemma \ref{ctvd-forbidden-char}, $G-X'$ contains an obstruction $O$ that induces a paw, or a diamond or a cycle of length at least 4. 
Note that this obstruction $O$ must contain $v$ which is a pendant vertex in $G$.
Therefore, $O$ is a paw but not a diamond or an induced cycle with at least four vertices.
Hence, $v$ is the pendant vertex in $O$.
Since $G - S$ has no { paw, no diamond, and no induced cycle with at least 4 vertices}, one of the vertices of $O$ is in $S$.
As $v$ is the pendant vertex of the paw $O$, it must be that $u$, unique neighbor of $v$ has two neighbors $z, w \in V(G - X')$ such that $zw \in E(G - X')$.
Since $v \notin S$, it must be that $z \in S$ or $w \in S$.
This contradicts our premise of the reduction rule that $u$ has no neighbor in $S$.
Hence, $X'$ be a (clique, tree)-deletion set of $G$.
\end{proof}

Let $C$ be a connected component of $G[V_2]$.
We say that $C$ is a {\em pendant tree} of $G[V_2]$ if either $C$ is a pendant vertex of $G$ or has a unique vertex $u$ that has a unique neighbor $v$ in $S$ { such that $uv$ is not a double-edge}, and no other vertex of $C$ has any neighbor in $S$.
We have the following observation.

\begin{observation}
\label{obs:unique-T-neighbor-property}
Let $C$ be a  connected component $C$ of $G[V_2]$.
Then, $C$ is a pendant tree if and only if $E(C, V(G) \setminus C)$ contains a single edge.
\end{observation}

\begin{proof}
Consider a connected component $C$ of $G[V_2]$.
Observe that $(C, V(G) \setminus C)$ is a cut of $G$.
If $C$ is a pendant tree of $G$, then there are two possibilities.
The first case is $C$ is a pendant vertex of $G$.
Then trivially $E(C, V(G) \setminus C)$ has a single edge.
The second case is $C$ is a pendant tree with at least two vertices.
Then, there is a unique vertex $u$ of $C$ such that $N_G(C) = N_G(u) \cap S = \{v\}$.
{ Additionally, as $C$ is a pendant tree, $uv$ cannot be a double-edge.}
Since no other vertex of $C$ has a neighbor in $S$ and $uv$ is the only edge in the cut $E(C, V(G) \setminus C)$ we have that $E(C, V(G) \setminus C)$ has a single edge.

Conversely, if $E(C, V(G) \setminus C)$ is a single edge, then there is a unique vertex $u \in C$ such that $u$ has only one neighbor $v$ in $S$ and no other vertex of $C$ has a neighbor in $S$.
{ Additionally, $uv$ is not a double-edge.}
In fact, no other vertex of $C$ has a neighbor in $V(G) \setminus C$.
Then, $C$ is a pendant tree of $G[V_2]$.
\end{proof}

Given a pendant tree $C$ of $G[V_2]$, we call $x \in S$ the {\em unique $S$-neighbor} of $C$ if $N_G(C) = \{x\}$ and $|N_G(x) \cap C| = 1$.
Furthermore, given a vertex $x \in S$, we call $C$ a {\em pendant tree-neighbor} of $x$ if $C$ is a pendant tree of $G[V_2]$ and $x$ is the {\em unique $S$-neighbor} of $C$.

\begin{reduction rule}
\label{rule:contract-pendant-tree-neighbor}
Let $C$ be a pendant tree in $G[V_2]$ such that $x \in S$ is a unique $S$-neighbor of $C$.
Then, delete all the vertices of $C$ except for the vertex $u$ that has the unique $S$-neighbor $x \in S$ in $C$ to obtain the graph $G'$.
{ The output instance is $(G', k)$}.
\end{reduction rule}

\begin{lemma}
\label{lemma:contract-tree-neighbor-safeness}
Reduction Rule \ref{rule:contract-pendant-tree-neighbor} is safe.
\end{lemma}


\begin{proof}
Let $C$ be a pendant tree in $G - S$ such that $u$ is the unique $S$-neighbor of $C$.
Let $x$ be the unique neighbor of $u$ in $C$.
Since $k' = k$, if $X$ is a solution to $(G, k)$, the forward direction ($\Rightarrow$) follows from the first item of Observation \ref{observation:subgraph-yes-instance}.

{ For the  backward direction ($\Leftarrow$), it suffices to prove that the vertices of $C \setminus \{u\}$ do not intersect any obstruction in $G$.
A very crucial observation is that no vertex of $C \setminus \{u\}$ participates in any cycle of $G$.
Consider any induced cycle of length at least 4.
Observe that the vertices of $C \setminus \{u\}$ do not participate in any induced cycle of length at least 4 in $G$.
Hence, the vertices of $C \setminus \{u\}$ is disjoint from every induced cycle of length at least 4 in $G$.
Consider a diamond of $G$.
Observe that every vertex of a diamond participates in a triangle in the diamond itself.
But, the vertices of $C \setminus \{u\}$ do not participate in any triangle of $G$.
Hence, the vertices of $C \setminus \{u\}$ are disjoint from every diamond of $G$.
Consider a paw $O$ of $G$.
If $O$ intersects a vertex $w$ from $C \setminus \{u\}$, then the neighbor of $w$ must be part of a triangle in $G$.
But every neighbor of $w$ is in $C$.
But no vertex of $C$ is part of any triangle in $G$.
Hence, the vertices of $C \setminus \{u\}$ is disjoint from every obstruction appearing in $G$.
Due to the second item of Observation \ref{observation:subgraph-yes-instance}, any solution of size at most $k$ to $(G', k)$ is a solution of size at most $k$ to $(G, k)$.
Therefore the reduction rule is safe.}
\end{proof}

\begin{lemma}
\label{lemma:pendant-tree-rule-not-applicable}
If Reduction Rule \ref{rule:contract-pendant-tree-neighbor} is not applicable to the input instance $(G, k)$, then any pendant tree of $G[V_2]$ is a pendant vertex of $G$.
\end{lemma}

\begin{proof}
Suppose that Reduction Rule \ref{rule:contract-pendant-tree-neighbor} is not applicable but $C$ is a pendant tree of $G[V_2]$ that is not a pendant vertex of $G$.	
Then, $C$ has a vertex $u$ such that $x \in S$ is the unique $S$-neighbor of $C$ in $G$.
In such a case, the premise of the Reduction Rule \ref{rule:contract-pendant-tree-neighbor} is satisfied implying that Reduction Rule \ref{rule:contract-pendant-tree-neighbor} is applicable to the instance $(G, k)$.
This leads to a contradiction to our assumption that Reduction Rule \ref{rule:contract-pendant-tree-neighbor} is not applicable.
\end{proof}

Recall the Definition~\ref{defn:flower} of $v$-flower.
Our next reduction rule uses the concept of $v$-flower and Proposition \ref{theorem:flower-cycle-hitting} as follows.

\begin{reduction rule}
\label{rr8}
For $v \in S$, we invoke Proposition \ref{theorem:flower-cycle-hitting} in $G[V_2 \cup \{v\}]$ with the integer $3k+1$.
If this gives a $v$-flower of order $3k+2$, then delete $v$ from $G$ and the new instance is $(G- \{v\}, k-1)$.
\end{reduction rule}

\begin{lemma} 
\label{lemma:flower-rule-safeness}
Reduction Rule \ref{rr8} is safe.
\end{lemma}

\begin{proof}
Clearly if $X$ is a solution in $G - \{v\}$, then $X \cup \{v\}$ is a solution in $G$.
Thus, if $(G- \{v\}, k-1)$ a yes-instance, then $(G,k)$ is also a yes-instance.
So, the backward direction ($\Leftarrow$) follows.

For the forward direction ($\Rightarrow$), let $X$ be a (clique, tree)-deletion set of $G$ of size at most $k$.
We claim that $v \in X$. 
Suppose not. 
Consider the set of cycles $C_1, C_2, \dots C_{3k+2}$ in the $v$-flower.
Consider those cycles from $\{C_1,\ldots,C_{3k+2}\}$ that are induced cycles of length at least 4 (holes) or double-edges. 
If there are at least $k+1$ such cycles, then at least one vertex from each of these cycles must be part of $X$. 
As $v \notin X$, due to Lemma \ref{ctvd-forbidden-char}, $X$ should contain at least one vertex from each of these $k+1$ cycles, which are disjoint apart from $v$.
This contradicts that $X$ is of size at most $k$. 
Thus, at most $k$ of these cycles are spanned by induced cycles length at least 4 and double edges.
Then, at least $3k + 2 - k = 2(k+1)$ of the cycles are triangles.
We arrange them arbitrarily into $k+1$ pairs and focus on one of the pairs of triangles $C = \{u_1,u_2,v\}$ and $C' = \{w_1,w_2,v\}$.
{ Since $G[V_2]$ is a forest and $u_1, u_2, w_1, w_2 \in V_2$, the graph $G[C \cup C']$ is not the clique $K_5$.
In addition, $G[\{u_1, u_2, w_1, w_2\}$ is a bipartite graph that has a maximum matching with edges $u_1 u_2$ and $w_1 w_2$.
Then, due to K{\H{o}}nig's Theorem, the bipartite graph $G[\{u_1, u_2, w_1, w_2\}]$ has vertex cover of size 2.
Hence, $G[\{u_1, u_2, w_1, w_2\}]$ has an independent set of size 2.
The independent set must contain one vertex from $\{u_1, u_2\}$ and one vertex from $\{w_1, w_2\}$.
Let $u_1$ and $w_1$ are pairwise nonadjacent.
Then, consider the vertices $G[\{u, u_1, u_2, w_1\}]$.
If $u_2 w_1 \in E(G)$, then $G[\{u, u_1, u_2, w_1\}]$ forms a diamond.
If $u_2 w_1 \notin E(G)$, then $G[\{u, u_1, u_2, w_1\}]$ forms a paw.
Hence, there is a paw or a diamond as an induced subgraph in $G[C \cup C']$.
By similar arguments, we can prove that if $u_1w_2 \notin E(G)$, or $u_2 w_1 \notin E(G)$, or $u_2 w_2 \notin E(G)$, then $G[C \cup C']$ contains a paw or a diamond as an induced subgraph.}
Thus, each of the $k+1$ pairs of triangles contains a graph that is a paw, or a diamond as induced subgraphs.
If $v \notin X$, then by Lemma  \ref{ctvd-forbidden-char}, there is a distinct vertex in $X$ from each of the $k+1$ pairs of cycles other than $v$. 
This again contradicts that $X$ is of size at most $k$ completing the proof of this lemma.
\end{proof}

Since Reduction Rule \ref{rr8} is not applicable, invoking Proposition \ref{theorem:flower-cycle-hitting} of order $(3k + 2)$ at $G[\{v\} \cup V_2]$ gives us a set $H_v \subseteq V_2$ that intersects all cycles of $G[\{v\} \cup V_2]$ that passes through $v$.
Our following lemma proves that the same vertex subset $H_v$ also intersects all paws and diamonds of $G[\{v\} \cup V_2]$ passing through $v$.

\begin{lemma}
\label{lemma:hitting-obstructions-containing-v}
If Reduction Rule \ref{rr8} is not applicable, then in polynomial time, we can obtain a vertex subset $H_v \subseteq V_2$ with  $|H_v| \leq 6k+4$ such that $H_v$ intersects every double-edge, every induced cycle of length at least 4, every paw, and every diamond in $G[\{v\} \cup V_2]$ that passes through $v$.
\end{lemma}

\begin{proof}
Since Reduction Rule \ref{rr8} is not applicable, graph $G[V_2 \cup \{v\}]$ has no $v$-flower of order $3k+2$.
It follows from the Proposition \ref{theorem:flower-cycle-hitting} that we obtain a set $H_v \subseteq V_2$ in polynomial time that intersects all cycles of $G[V_2 \cup \{v\}]$ that pass through $v$.
Therefore, $H_v$ trivially intersects all double edges of $G[\{v\} \cup V_2]$ and all induced cycles of $G[\{v\} \cup V_2]$ containing at least 4 vertices.

Consider a diamond $O$ in $G[\{v\} \cup V_2]$ such that $v \in O$.
Observe that $v$ must have at least two neighbors in $O \cap V_2$ and $|O \cap V_2| = 3$.
Since $G[V_2]$ is a forest, the subgraph $G[O \cap V_2]$ has no cycle.
Hence, the only way $O$ can form a diamond that $G[V_2 \cap O]$ is an induced $P_3$ and $v$ is adjacent to all other vertices of $O \cap V_2$.
Then, there are two triangles and a $C_4$ as subgraph in $G[\{v\} \cup V_2]$ that is in $G[O]$.
Since $H_v$ intersects both the triangles and this subgraph $C_4$, it must be that $H_v \cap O \neq \emptyset$.

Now, we consider a paw $O$ of $G[\{v\} \cup V_2]$ such that $v \in O$.
Since $O \cap V_2$ cannot have any cycle, it must be that the pendant vertex of this paw $O$ is in $V_2$.
Then, the only way for $O$ to form a paw containing $v$ is that $v$ is adjacent to two pairwise adjacent vertices $x_1, x_2$ of $V_2$ and $x_3$ is pendant vertex of $O$ such that $O = \{v, x_1, x_2, x_3\}$.
But, this provides us a cycle $\{v, x_1, x_2\}$ that passes through $v$.
As $H_v$ intersects $\{x_1, x_2\}$, it follows that $H_v$ intersects $O$ that passes through $v$.

This completes the proof that $H_v$ intersects all double-edges, all induced cycles of length at least 4, all paws, and all diamonds of $G[\{v\} \cup V_2]$ that pass through $v$.
\end{proof}


\paragraph{Construction of Auxiliary Bipartite Graph:}
If none of the above reduction rules are applicable, we exploit the structural properties of the graph using the above mentioned lemmas and construct an auxiliary bipartite graph that we use in some reduction rules later.
Let $\mathcal{C}$ denote the connected components of $G[V_2 \setminus H_v]$ that are adjacent to $v$.
In other words, if $D \in \mathcal{C}$, then $D$ is a connected component of $G[V_2 \setminus H_v]$ such that $v$ is adjacent to $D$.
We formally construct the following auxiliary bipartite graph $\mathcal{H}$ as follows.
\begin{definition}\label{definition:aux-bipartite}
Given $v \in S$, let $H_v$ denote the set of at most $6k+4$ vertices obtained by Lemma \ref{lemma:hitting-obstructions-containing-v} to the graph $G[\{v\} \cup V_2]$.
Consider the graph $G[V_2 \setminus H_v]$ that is a forest.
We define an auxiliary bipartite graph $\mathcal{H}_v = (H_v \cup (S \setminus \{v\}), \mathcal{C})$ where $H_v \cup (S \setminus \{v\})$ is on one side, and 
$\mathcal{C}$ 
 on the other side.
The set $\mathcal{C}$ contains a vertex for each connected component $C$ of $G[V_2 \setminus H_v]$ such that $C$ has a vertex adjacent to $v$.
We add an edge between $h \in H_v \cup (S \setminus \{v\})$, and connected component $C \in \mathcal{C}$ if $h$ is adjacent to a vertex in component $C \in \mathcal{C}$. 
\end{definition}

We prove the following observation that is crucial to the next reduction rule which upper bounds the number of edges incident to a vertex $v \in S$ with other endpoint being in $V_2$.

%
\begin{observation}\label{observation:lower-bound-C-aux-bip}
Let $\mathcal{H}_v = (H_v \cup (S \setminus \{v\}), \mathcal{C})$ be the auxiliary bipartite graph as defined in Definition \ref{definition:aux-bipartite}.
If $v$ has degree more than $52(k + 1)$ in $G[\{v\} \cup V_2]$, then
\begin{enumerate}[(i)]
	\item\label{caligraphic-C-size} $\mathcal{C}$ has more than $4(|S| +|H_v|)$ components, and
	\item\label{caligraphic-C-v-connectivity} every connected component $C$ of $\mathcal{C}$ has exactly one vertex $x$ adjacent to $v$ and $vx$ is not a double-edge.
\end{enumerate}  
\end{observation}


\begin{proof}
By hypothesis, there are more than $52k + 52$ edges that are incident to $v$ with the other endpoint in $V_2$.
Note that $|H_v| \leq 6k+4$, and for every $w \in S \setminus \{v\}$.
{ As Reduction Rule \ref{rule:multiplicity-reduction} is not applicable, there can be at most two edges between $v$ and $w$}.
In addition, as Reduction Rule \ref{rule:high-degree-double-edge}, there are at most $k$ vertices $w \in V_2$ such that there is a double edge between $v$ and $w$.
In fact, if $vw$ is a double-edge and $w \in V_2$, then due to Lemma \ref{lemma:hitting-obstructions-containing-v}, $w \in H_v$.
Hence, there are at most $k$ vertices $w \in H_v \cup (S \setminus \{v\})$ such that there is a double edge between $v$ and $w$. 
As $|H_v \cup (S \setminus \{v\})| \leq 10k + 4$ and at most $k$ of these vertices are adjacent to $v$ with a double edge, therefore, there are at most $12k + 4$ edges are incident to $v$ with other endpoint being in $H_v \cup (S \setminus \{v\})$.
Due to Lemma \ref{lemma:hitting-obstructions-containing-v}, $H_v$ intersects every cycle, every paw, every diamond, and every double-edge of $G[\{v\} \cup V_2]$ passing through $v$.
It implies that there does not exist any $z \in V_2 \setminus H_v$ that is adjacent to $v$ by a double edge.
In addition, there cannot exist any cycle $C^*$ containing vertices only from $v$ and some connected component $C \in \CC$.
Therefore, $v$ is adjacent to at most one vertex in each connected component $C$ of $\mathcal{C}$.
In particular, if one such vertex $x$ exists in $C$, then $xv$ is not a double-edge.
Additionally, due to Definition \ref{definition:aux-bipartite}, every $C \in \mathcal{C}$ has at least one vertex $x$ such that $vx$ is an edge.
Hence, for every $C \in \mathcal{C}$, there is exactly one vertex $x \in C$ such that $xv \in E(G)$ and $xv$ is not a double-edge.
This completes the proof for the last item.
{Finally, note that $v$ is adjacent to at least $52k + 52 - (12k + 4) = 40k + 48$ vertices other than $H_v$ in $G[\{v\} \cup V_2]$.}
Hence, the number of connected components in $\mathcal{C}$ is at least $40k + 48$ implying that $|\mathcal{C}| \geq 40k + 44 > 4(10k + 4) = 4(|S| + |H_v|)$. 
This completes the proof of the first item.
\end{proof}




\paragraph{Applying the New Expansion Lemma:}
If there is a vertex $v \in S$ such that there are at least $52(k + 1)$ edges incident to $v$ with the other endpoints being in $V_2$, then Observation \ref{observation:lower-bound-C-aux-bip} implies that $|\mathcal{C}| > 4(|S| + |H_v|)$.
Suppose that we apply new $4$-expansion lemma (Lemma \ref{lemma:new-expansion-lemma} with $q=4$) on $\mathcal{H}_v$ to obtain $A \subseteq (S \setminus \{v\}) \cup H_v, \BB \subseteq \CC$ with a $4$-expansion $\widehat M$ of $A$ into $\BB$.
Then, it satisfies that (i) $|\CC \setminus \BB| \leq 4|((S \setminus \{v\}) \cup H_v) \setminus A|$ and $N_{\mathcal H}(\BB) \subseteq A$.
As $|\CC \setminus \BB| \leq 4|((S \setminus \{v\}) \cup H_v) \setminus A|$ and $|\CC| > 4(|S| + |H_v|)$, it must be that $|{\BB}| > 4|A|$.
Then, there must be a component $C^* \in \BB$ such that $C^*$ is not an endpoint of $\widehat M$ (or not saturated by $\widehat M$).
Let $\widehat{\BB} \subseteq \BB$ denote the components of $\BB$ that are saturated by $\widehat M$.
As some component of $\BB$ is not in $\widehat{\BB}$, it must be that $\widehat{\BB} \subset \BB$.
We use these characteristics crucially to prove that our next reduction rule is safe.

\begin{lemma}
\label{lemma:A-should-be-nonempty}
Let $(G, k)$ be an input instance to which Reduction Rules \ref{rr1}-\ref{rr8} are not applicable and consider $v \in S$ with degree at least $52(k+1)$ in $G[V_2 \cup \{v\}]$.
Next, we construct the auxiliary bipartite graph $\mathcal{H}_v = (H_v \cup (S \setminus \{v\}), \CC)$ as per Definition \ref{definition:aux-bipartite} and subsequently invoke Lemma \ref{lemma:new-expansion-lemma} (i.e. new $q$-expansion lemma) with $q = 4$ and compute the sets $A \subseteq (H_v \cup S) \setminus \{v\}$ and ${\BB} \subseteq {\CC}$ such that
\begin{itemize}
	\item there is a $4$-expansion $\widehat{M}$ of ${A}$  onto ${\BB}$ in $\mathcal{H}_v$, and
	\item $N_{\mathcal{H}_v}({\BB}) \subseteq A$.
\end{itemize} 
Then, the following statements are satisfied.
\begin{enumerate}[(i)]
	\item\label{expansion-A-nonempty} ${A} \neq \emptyset$,
	\item\label{expansion-unsaturated-C-exists} there is a component $C \in {\BB}$ that is not saturated by $\widehat{M}$, and
	\item\label{expansion-C-neighbors} for any component $C$ of $\BB$, $N_G(C) \subseteq A \cup \{v\}$.
\end{enumerate} 
\end{lemma}

\begin{proof}
Suppose that the premise of the lemma are satisfied and we obtain the sets $A \subseteq (H_v \cup S) \setminus \{v\}$, ${\BB} \subseteq {\CC}$ and a 4-expansion $\widehat{M}$ of $A$ onto $\BB$.
As $v$ has degree more than $52(k+1)$ in $G[V_2 \cup \{v\}]$, due to item-(\ref{caligraphic-C-size}) of Observation \ref{observation:lower-bound-C-aux-bip} that $|\CC| > 4(|S| + |H_v|)$.
Then due to the second item (ii) of Lemma \ref{lemma:new-expansion-lemma}, it follows that $|\CC \setminus \BB| \leq 4|((H_v \cup S) \setminus \{v\}) \setminus A)|$.
Then, $|\BB| > 4|A|$, implying that $\BB \neq \emptyset$.

First, we prove that $A \neq \emptyset$.
Now, we assume for the sake of contradiction that $A = \emptyset$.
As $N_{\mathcal{H}_v}(\BB) \subseteq A$, it implies that the components of $\BB$ have no neighbor in $(H_v \cup S) \setminus \{v\}$.
Comparing the values that $|\CC \setminus \BB| \leq 4|((H_v \cup S) \setminus \{v\}) \setminus A)|$ and $|\CC| > 4|((H_v \cup S) \setminus \{v\}) \setminus A)| + 4$, there are 4 components $D_1, D_2, D_3, D_4$ in $\BB$ that are not adjacent to any vertex of $(H_v \cup S) \setminus \{v\}$.
But all these components $D_1, D_2, D_3, D_4$ have a unique vertex $u$ adjacent to $v$ such that $uv$ is not a double-edge.
Then, observe that every component of $\BB$ is a component of $G[V_2]$ and is a pendant tree in $G[V_2]$.
Additionally, $v$ is the unique $S$-neighbor of each component of $\BB$.
Then, as Reduction Rule \ref{rule:contract-pendant-tree-neighbor} is not applicable each component of $\BB$ is a component of $G[V_2]$ with a single vertex.
It implies that each of $D_1, D_2, D_3, D_4$ are pendant vertices of $G$.
This implies that Reduction Rule \ref{rr1} is applicable that leads to a contradiction to our assumption in the premise.
Hence, $A \neq \emptyset$, completing the proof of item-(\ref{expansion-A-nonempty}).

Now, we prove the item-(\ref{expansion-unsaturated-C-exists}).
Due to item-(\ref{expansion-A-nonempty}), $A \neq \emptyset$.
It also holds that $|\BB| > 4|A|$, it implies that exactly $4|A|$ components from $\BB$ are saturated by $\widehat{M}$.
Then, clearly, some component $C$ exists in $\BB$ that is unsaturated by $\widehat{M}$, completing the proof of item-(\ref{expansion-unsaturated-C-exists}).
Finally, consider any component $C \in \BB$.
As $\BB \subseteq \CC$, $C$ has a vertex that is adjacent to $v$.
As $N_{\mathcal{H}_v}(\BB) \subseteq A$, and $C \in \BB$, there cannot exist any neighbor of $C$ in $G$ that is outside $A \cup \{v\}$.
Hence, $N_G(C) \subseteq A \cup \{v\}$, completing the item-(\ref{expansion-C-neighbors}).
\end{proof}

Now, we state the next reduction rule and subsequently prove that it is safe.

\begin{reduction rule}
\label{reduction-rule-expansion-lemma-trees}
Let $v \in S$ be a vertex with degree at least $52(k + 1)$ in $G[\{v\} \cup V_2]$ and let $\mathcal{H}_v$ be the auxiliary bipartite graph as illustrated in Definition \ref{definition:aux-bipartite}.
We invoke the algorithm provided by Lemma \ref{lemma:new-expansion-lemma} (i.e. new $q$-expansion lemma with $q=4$) to compute sets $A \subseteq H_v \cup (S \setminus \{v\})$ and $\BB \subseteq \CC$ such that
\begin{itemize}
	\item  $A$ has a $4$-expansion $\widehat{M}$ into $\BB$ in $\mathcal{H}_v$, and
	\item $N_{\mathcal{H}_v}(\BB) \subseteq A$.
\end{itemize}
 
Let $\widehat{\BB} \subseteq \BB$ denote the vertices of $\BB$ that are saturated by $\widehat{M}$ (endpoints of $\widehat{M}$ in $\BB$).
Remove the edges between $v$ and the connected components in $\widehat{\BB}$ in $G$ and create a double edge between $v$ and every vertex of $A$ to obtain the graph $G'$.
The new instance is $(G', k')$ with $k = k'$.
We refer to the Figure \ref{fig:new-expansion-updated-rule} for an illustration.
\end{reduction rule}

Before we prove the safeness of the above reduction rule, we prove the following two lemmas.

\begin{lemma}
\label{lemma:solution-properties-with-H}
Let $X$ be an optimal (clique, tree)-deletion set of $G$ of size at most $k$ and $A, \BB, \widehat{\BB}$ denote the vertex subsets obtained from Reduction Rule \ref{reduction-rule-expansion-lemma-trees}.
Then, $v \in X$ or $A \subseteq X$.
\end{lemma}

\begin{proof}
Let $X$ be an optimal (clique, tree)-deletion set of $G$ of size at most $k$.
{ Note that we have already invoked Lemma \ref{lemma:new-expansion-lemma} (new $q$-expansion lemma) with $q = 4$ in Reduction Rule \ref{reduction-rule-expansion-lemma-trees}.
The sets obtained are $A \subseteq (S \cup H_v) \setminus \{v\}$ and $\BB \subseteq \CC$ such that $N_{\mathcal{H}_v}(\BB) \subseteq A$ (due to item-(ii) of Lemma \ref{lemma:new-expansion-lemma})}.
As { the} degree of $v$ in $G[V_2 \cup \{v\}]$ is at least $52(k+1)$, due to item-(\ref{expansion-A-nonempty}) of Lemma \ref{lemma:A-should-be-nonempty}, it holds that $A \neq \emptyset$.
Additionally, due to item-(\ref{caligraphic-C-v-connectivity}) of Observation \ref{observation:lower-bound-C-aux-bip}, every component $C \in \CC$ has exactly one vertex that is adjacent to $v$.

Suppose for the sake of contradiction that there is an optimal (clique, tree)-deletion set $X^*$ of $G$ of size at most $k$ such that $v \notin X^*$ and $A \not\subseteq X^*$.
Let $X_{\BB}^*$ denote the intersection of $X^*$ with the vertices that are in the connected components in $\BB$.
We set $\widehat{X} = (X^* \setminus X^*_{\BB}) \cup (A \setminus X^*) \cup \{v\}$. 
We claim that $\widehat{X}$ is a (clique, tree)-deletion set of $G$ and $|\widehat{X}| < |X^*|$.

First, we prove that $\widehat{X}$ is a (clique, tree)-deletion set of $G$.
For every component $C$ of $\BB$, it holds due to item-(\ref{expansion-C-neighbors}) of Lemma \ref{lemma:A-should-be-nonempty} that $N_G(C) \subseteq A \cup \{v\}$.
By construction of $\widehat{X}$, as $A \cup \{v\} \subseteq \widehat{X}$.
Consider a double-edge of $G$ between $x$ and $y$.
Every double-edge must be incident to some vertex of $S$.
If $x$ or $y$ is in $A \cup \{v\}$, then $\widehat{X}$ intersects the double-edge between $x$ and $y$.
Otherwise, consider the double-edges that are not incident to any vertex of $A \cup \{v\}$.
Then, those double-edges are incident to the vertices of $S \setminus (A \cup \{v\})$.
Due to item-(\ref{expansion-C-neighbors}) of Lemma \ref{lemma:A-should-be-nonempty}, for every component $C \in \BB$, it holds that $N_G(C) \subseteq A \cup \{v\}$.
If a double-edge between $x$ and $y$ being incident to $C \in \BB$, then neither $x$, nor $y$ can be part of $S \setminus (A \cup \{v\})$.
Hence, a double-edge incident to the vertices of $S \setminus (A \cup \{v\})$ cannot be incident to any component $C \in \BB$.
Hence, those double-edges are incident to the components outside $\BB$ and incident to $S \setminus (A \cup \{v\})$.
As those vertices from $X^*$ are not deleted while constructing $\widehat{X}$, and $X^*$ intersects a double-edge between $x$ and $y$, it implies that $\{x, y\} \cap \widehat{X} \neq \emptyset$.
Hence, all double-edges are intersected by $\widehat{X}$.


Let us now consider the other obstructions, paw, diamond, and induced cycles of length at least 4 of $G$.
If those obstructions are not intersected by any component $C \in \BB$, then such obstructions are already intersected by $X^*$.
As those vertices from $X^*$ appearing in the obstructions disjoint from $C$ are not deleted when constructing $\widehat{X}$, the obstructions disjoint from $C$ are intersected by $\widehat{X}$.
So, we consider those paws, diamonds, induced cycles of length at least 4 that intersect a component $C \in \BB$.
As $C$ is a tree, to form any such obstruction, paw, diamond, or induced cycles of length at least 4 in $G$, it must be that the obstruction contains some vertex from $N_G(C)$.
Due to item-(\ref{expansion-C-neighbors}) of Lemma \ref{lemma:A-should-be-nonempty}, $N_G(C) \subseteq A \cup \{v\}$.
As $A \cup \{v\} \subseteq \widehat{X}$, the obstructions that are paws, diamonds, induced cycles of length at least 4 and intersect $C \in \BB$ are also intersected by $\widehat{X}$.
Hence, $\widehat{X}$ is a (clique, tree)-deletion set of $G$.

Now, we claim that $|\widehat X| < |X^*|$.
Since $A \not\subseteq X^*$ and $v \notin X^*$, there is $x \in A \setminus X^*$.
Due to Lemma \ref{lemma:new-expansion-lemma}, there are four connected components $D_1, D_2, D_3, D_4$ in $\widehat{\BB}$ such that $v$ is adjacent to one vertex from each of $D_1, D_2, D_3, D_4$ and $x$ is adjacent to each of $D_1, D_2, D_3, D_4$ in $\mathcal{H}_v$.
If $v$ is adjacent to $x$, then $G[\{v, x\} \cup D_1 \cup D_2 \cup D_3 \cup D_4]$ has several diamonds or subdivision of diamonds that contains $\{x, v\}$ and vertices from exactly two connected components from $\{D_1, D_2, D_3, D_4\}$.
In particular, for every pair $i, j \in \{1, 2, 3, 4\}$, there is a diamond or a subdivision of diamond that contain $\{x, v\}$, and vertices from $D_i$ and $D_j$.
{ Consider an induced subgraph that is a proper subdivision of a diamond, i.e. at least one edge is subdivided by at least one vertex.
Then, note that this procedure converts a triangle into an induced cycle of length at least 4.
This justifies the existence of an induced cycle with at least 4 vertices.
Therefore, for every pair $i, j \in [4]$, there is a diamond or an cycle of length at least 4 as an induced subgraph.}
Since $v, x \notin X^*$, it must be that $X^*$ must have at least one vertex from at least three of these connected components $\{D_1, D_2, D_3, D_4\}$. 
As our construction procedure of $X^*$ to $\widehat{X}$ removes the vertices of the components of $\BB$ from $X^*$ and adds the vertices of $A \setminus X^*$, it follows that for every $x \in A \setminus X^*$, one vertex is added and at least three additional vertices appearing in the components of $\{D_1, D_2, D_3, D_4\}$ are removed from $X^*$.
Finally, after adding $v$, it ensures that $\widehat{X}$ has strictly lesser vertices than $X^*$ contradicting the optimality of $X^*$.
This completes the proof.
\end{proof}

\begin{figure}[t]
\centering
	\includegraphics[scale=0.28]{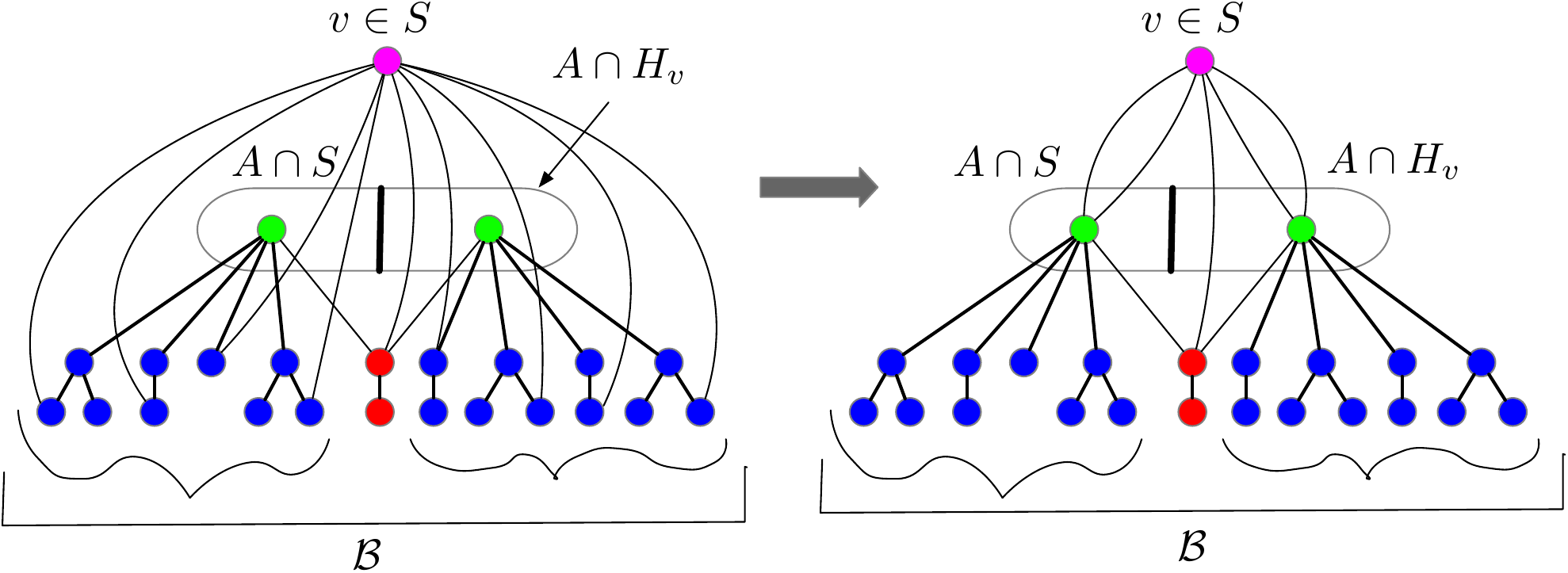}
	\caption{An illustration of Reduction Rule \ref{reduction-rule-expansion-lemma-trees}. The blue components are the components of $\BB$ and the red component is the chosen component $C$ for which the edge between $v$ and $C$ is not deleted. The blue components of $\BB$ are the ones that are the endpoints of expansion $\widehat M$.}
\label{fig:new-expansion-updated-rule}
\end{figure}

We use the previous lemmata and observations to prove that  Reduction Rule \ref{reduction-rule-expansion-lemma-trees} is safe.

\begin{lemma}
\label{lemma:reduction-rule-expansion lemma-trees-safeness}
Reduction Rule \ref{reduction-rule-expansion-lemma-trees} is safe.
\end{lemma}

\begin{proof}
For the forward direction ($\Rightarrow$), let $X$ be a (clique, tree)-deletion set of $G$ with at most $k$ vertices.
We assume without loss of generality that $X$ is an optimal (clique, tree)-deletion set of $G$.
Due to Lemma \ref{lemma:solution-properties-with-H}, it follows that $v \in X$ or $A \subseteq X$ or both.
For each of those cases, observe that $G' - X$ is an induced subgraph of $G - X$. 
Hence, $X$ is a (clique, tree)-deletion set of $G'$ of size at most $k'$.

For the backward direction ($\Leftarrow$), let $X'$ be a (clique, tree)-deletion set of $G'$ of size at most $k' (= k)$.
Since there is a double-edge between $v$ and every vertex of $A$, it implies that $A \subseteq X'$ or $v \in X'$.
In case $v \in X'$, then the graphs $G - X'$ and $G' - X'$ are precisely the same.
Therefore, $X'$ is a (clique, tree)-deletion set of $G$ as well.
For the other case, we have that $A \subseteq X'$ but $v \notin X'$.
Suppose for the sake of contradiction that $G - X'$ has a double-edge, or some component of $G - X'$ is neither a clique, nor a tree.
Then, $G - X'$ has an obstruction $O$.
Due to the item-(\ref{expansion-C-neighbors}) of Lemma \ref{lemma:A-should-be-nonempty}, it holds that $N_G(C) \subseteq A \cup \{v\}$.
By construction of $G'$, an edge of $G$ is not an edge of $G'$ only if one endpoint $u$ is in some connected component $C' \in \widehat{\BB}$ and the other endpoint is $v$.
The obstruction $O$ in $G - X'$ must contain such an edge $uv$.
Due to Observation \ref{observation:lower-bound-C-aux-bip}, $uv$ cannot be a double-edge.
Hence, we can safely assume that $G - X'$ has no double-edge.
{ Also, $N_G(C) \subseteq A \cup \{v\}$ for any $C \in \BB$}, and $A \subseteq X'$, the only possible way $uv$ edge can be part of an obstruction $O$ in $G - X'$ is a paw.
{ As $uv$ is not a double-edge, this obstruction $O$ is a paw that contains $u$ and $v$}.
Now, what is crucial here is that due to item-(\ref{expansion-unsaturated-C-exists}) of Lemma \ref{lemma:A-should-be-nonempty}, there is a component $C \in \BB$ such that $C$ is not saturated by $\widehat{M}$.
Hence, such a component $C \in \BB \setminus \widehat{\BB}$.
{ Furthermore, Reduction Rule \ref{reduction-rule-expansion-lemma-trees} deletes the edges that have one endpoint in $C \in \widehat{\BB}$ and the other endpoint is $v$}.
Hence, Reduction Rule \ref{reduction-rule-expansion-lemma-trees} has chosen not to delete the edge $vu^*$ such that  $u^* \in C$, $vu^* \in E(G)$ and $C \in \BB \setminus \widehat{\BB}$.
But by construction of $G'$, $vu^* \in E(G')$.
Observe that $O^* = (O \setminus \{u\}) \cup \{u^*\})$ also induces a paw in the graph $G$.
Then, $u^*$ must be in the set $X'$ as otherwise it would contradict that $X'$ is a (clique, tree)-deletion set of $X'$.
So, we set $X^* = (X' \setminus \{u^*\}) \cup \{v\}$ and clearly by construction $|X^*| = |X'|$.
As the item-(\ref{expansion-C-neighbors}) of Lemma \ref{lemma:A-should-be-nonempty}, the neighborhood of any connected component of $\BB$ is contained in $A \cup \{v\}$, this ensures us that $X^*$ is a (clique, tree)-deletion set of $G$.
This completes the proof of the lemma.
\end{proof}

Observe that the above mentioned reduction rules does not increase the degree of $v$.
When none of the above mentioned reduction rules are applicable, no connected component $C$ (with at least three vertices) of $G[V_2]$ can have two leaves $u$ and $v$ both of which are pendant vertices in $G$.
But, it is possible that $C$ has one leaf $u$ that is a pendant vertex of the whole graph $G$.


\begin{lemma}
\label{lemma:trees-bound}
Let $G$ be the graph obtained after applying Reduction Rules \ref{rr1}-\ref{reduction-rule-expansion-lemma-trees} exhaustively. 
Then the number of vertices in $V_2$ is at most $1325k|S|$.
\end{lemma}

\begin{proof}
We use $N$ to denote the vertices of $V_2$ that are adjacent to some vertex of $S$.
As Reduction Rule \ref{reduction-rule-expansion-lemma-trees} is not applicable, for every $v \in S$, there are at most $52(k + 1) (\leq 53k)$ edges incident to $v$ with the other endpoint being in $V_2$ (hence in $N$).
Hence, the number of vertices of $N$ is at most $53k|S|$.

Let us bound the number of leaves in the forest $G[V_2]$ that is not in $N$.
Since Reduction Rule \ref{reduction-rule-leaf-neighbor-not-adj} is exhaustively applied, such a leaf is adjacent to a vertex that is adjacent to some vertex in $S$. 
Since Reduction Rule \ref{rr4} is exhaustively applied, two such leaves are not adjacent to the same vertex.
Hence, we can define an injective function from the leaves of $G[V_2 \setminus N]$ to the internal vertices in the forest $G[V_2]$ that is in $N$.

Thus, the total number leaves in $G[V_2]$ is at most $2|N|$.
Since, the number of vertices with degree at least 3 in $G[V_2]$ is at most the number of leaves in $G[V_2]$, the number of vertices of $G[V_2]$ with degree at least three is at most $2|N|$.
Therefore, the sum of the number leaves in $G[V_2]$, and the number of vertices with degree at least 3 of $G[V_2]$ is at most $4|N|$.
Additionally, there are some vertices of $N$ that are neither counted as a leaf of $G[V_2]$ nor is counted as a vertex of degree at least three in $G[V_2]$.
{The number of such vertices} is at most $|N|$.

{It remains to consider the degree-2 vertices in $G[V_2]$ that are also degree-2 vertices of $G$.
Let $P$ be path in $G[V_2]$ whose endpoints are leaves of $G[V_2]$, degree at least 3 in $G[V_2]$ or in $N$. Note that the internal vertices of $P$ are degree-2 vertices in $G$. If one of the endpoints is a pendant vertex in $G$, since Reduction Rule \ref{rr5} is exhaustively applied, the size of the path is at most $2$. Else if both of the endpoints have degree at least 3 in $G$, since Reduction Rule \ref{rr7} is exhaustively applied, the size of the path is at most $4$. Hence both the end points have degree exactly 2 in $G$, and thus are in $N$. If both of the endpoints have a common neighbor, then as Reduction Rule \ref{rule:pendant-cycle} is not applicable, the size of the path is at most $2$. Thus, both the endpoints have different neighbors.
Then there exists a degree-2 overbridge $P'$ in $G$ such that $P$ is contained in $P'$. As  Reduction Rule \ref{rr7} is not applicable, the size of $P'$ (and thus $P$) is at most $4$. Thus, we conclude that the size of $P$ is at most $4$. If we replace all such paths $P$ in $G[V_2]$ by edges}, we get a forest $H$ with at most $5|N|$ vertices, and thus edges.
Since each edge of $H$ corresponds to at most 4 vertices, we have at most $20|N|$ vertices each of which have degree exactly two.

Thus the total number of vertices in $G[V_2]$ is bounded by $25|N|$.
Since $|N| \leq 53k|S|$, the total number of vertices in the forest $G[V_2]$ is bounded by $1325k|S|$.
\end{proof}

Combining Lemma \ref{lemma:cliques-bound} and Lemma \ref{lemma:trees-bound}, we are ready to prove our final result that we restate below.

{\MainResult*}

\begin{proof}
Given the input instance $(G, k)$, the kernelization algorithm invokes Reduction Rules \ref{rr1}-\ref{reduction-rule-expansion-lemma-trees} exhaustively.
Let $(G', k')$ denote the output instance such that $S'$ is a (clique, tree)-deletion set of $G'$ with at most $4k'$ vertices.
When a reduction rule is applied, the parameter $k$ never increases and no vertex is added to $S$.
Hence, $k' \leq k$ and $|S'| \leq |S| \leq 4k$.
Suppose that $V_1 \subseteq G' - S'$ denotes the vertices such that every connected component of $G[V_1]$ is a clique and $V_2 \subseteq G' - S'$ denotes the vertices such that every connected component of $G[V_2]$ is a tree.
Since Reduction Rules \ref{rr1}-\ref{rr6} are not applicable, it follows from Lemma \ref{lemma:cliques-bound} that $|V_1|$ is $\OO(k^5)$. 
Additionally, as Reduction Rules \ref{rr1}-\ref{reduction-rule-expansion-lemma-trees} are not applicable, it follows from Lemma \ref{lemma:trees-bound} that $|V_2| \leq 1325k|S'| = 5300k^2$. 
Therefore, the total number of vertices in $G'$ is $|S'| + |V_1| + |V_2|$  which is $\OO(k^5)$.
\end{proof}

\section{Conclusions and Future Research}
\label{sec:conclusion}

Our paper initiates a study of polynomial kernelization for vertex deletion to pairs of scattered graph classes.
Jacob et al. \cite{JacobMR23jcss} have provided a $4^k n^{\OO(1)}$-time algorithm for {\ctvd} and it was subsequently improved to a $3.46^k n^{\OO(1)}$-time algorithm by Tsur \cite{TSUR2025Kernel}.
It would also be interesting to design an FPT algorithm where the base of the exponent is (substantially) improved to $3$.
The asymptotically fastest randomized algorithm runs in $3.13^k n^{\OO(1)}$-time \cite{TSUR2025Kernel}.
Therefore, designing a randomized algorithm that runs in $3^k n^{\OO(1)}$-time is also an open problem.
On a broader level, it would be interesting to explore the possibility of getting a polynomial kernel for problems where the objective is to delete a set of at most $k$ vertices so that the connected components would belong to other interesting pairs of graph classes, such as (interval graph, trees), (chordal graph, bipartite permutation), (interval, cactus) etc.

In addition, vertex/edge deletion to scattered graph classes are also interesting from approximation algorithms perspective.
In fact, it would be interesting to improve the approximation guarantee of Proposition \ref{lemma-ctvd-approximation} that is also an open problem.
Additionally, the dual version of this problem, i.e. ``packing vertex-disjoint obstructions to the scattered class of cliques and trees'' is also interesting from the perspective of parameterized complexity. 
The same problem can be considered as packing vertex-disjoint induced subgraphs that are paws or diamonds or cycles of length at least 4.
A natural approach to solve this problem requires to design an Erdos-Posa style theorem for packing obstructions for scattered class of cliques and trees.
Finally, a more general open problem is to identify pairs of graph classes $(\Pi_1, \Pi_2)$ for which vertex deletion to $\Pi_1$ as well as vertex deletion to $\Pi_2$ admits polynomial sized kernels, but $(\Pi_1, \Pi_2)$-Deletion does not admit a polynomial kernel.



\begin{thebibliography}{10}

\bibitem{Faisal2010}
Faisal~N. Abu-Khzam.
\newblock A kernelization algorithm for d-hitting set.
\newblock {\em J. Comput. Syst. Sci.}, 76(7):524--531, 2010.

\bibitem{BabuJKR24}
Jasine Babu, Ajay~Saju Jacob, R.~Krithika, and Deepak Rajendraprasad.
\newblock Packing arc-disjoint cycles in oriented graphs.
\newblock {\em J. Comput. Syst. Sci.}, 143:103530, 2024.

\bibitem{BonamyDFJP19}
Marthe Bonamy, Konrad~K. Dabrowski, Carl Feghali, Matthew Johnson, and
  Dani{\"{e}}l Paulusma.
\newblock Independent feedback vertex set for p\({}_{\mbox{5}}\)-free graphs.
\newblock {\em Algorithmica}, 81(4):1342--1369, 2019.

\bibitem{BoralCKP16}
Anudhyan Boral, Marek Cygan, Tomasz Kociumaka, and Marcin Pilipczuk.
\newblock A fast branching algorithm for cluster vertex deletion.
\newblock {\em Theory Comput. Syst.}, 58(2):357--376, 2016.

\bibitem{BroersmaFGKPP15}
Hajo Broersma, Jir{\'{\i}} Fiala, Petr~A. Golovach, Tom{\'{a}}s Kaiser,
  Dani{\"{e}}l Paulusma, and Andrzej Proskurowski.
\newblock Linear-time algorithms for scattering number and
  hamilton-connectivity of interval graphs.
\newblock {\em J. Graph Theory}, 79(4):282--299, 2015.

\bibitem{Cai96}
Leizhen Cai.
\newblock Fixed-parameter tractability of graph modification problems for
  hereditary properties.
\newblock {\em Inf. Process. Lett.}, 58(4):171--176, 1996.

\bibitem{CaoM15}
Yixin Cao and D{\'{a}}niel Marx.
\newblock Interval deletion is fixed-parameter tractable.
\newblock {\em {ACM} Trans. Algorithms}, 11(3):21:1--21:35, 2015.

\bibitem{CaoM16}
Yixin Cao and D{\'{a}}niel Marx.
\newblock Chordal editing is fixed-parameter tractable.
\newblock {\em Algorithmica}, 75(1):118--137, 2016.

\bibitem{Chen16d}
Jianer Chen.
\newblock Vertex cover kernelization.
\newblock In {\em Encyclopedia of Algorithms}, pages 2327--2330. 2016.

\bibitem{ChenKX10}
Jianer Chen, Iyad~A. Kanj, and Ge~Xia.
\newblock Improved upper bounds for vertex cover.
\newblock {\em Theor. Comput. Sci.}, 411(40-42):3736--3756, 2010.

\bibitem{chudnovsky2006strong}
Maria Chudnovsky, Neil Robertson, Paul Seymour, and Robin Thomas.
\newblock The strong perfect graph theorem.
\newblock {\em Annals of mathematics}, pages 51--229, 2006.

\bibitem{CormenLRS09}
Thomas~H. Cormen, Charles~E. Leiserson, Ronald~L. Rivest, and Clifford Stein.
\newblock {\em Introduction to Algorithms, 3rd Edition}.
\newblock {MIT} Press, 2009.

\bibitem{Courcelle90}
Bruno Courcelle.
\newblock The monadic second-order logic of graphs. i. recognizable sets of
  finite graphs.
\newblock {\em Inf. Comput.}, 85(1):12--75, 1990.

\bibitem{CyganFKLMPPS15}
Marek Cygan, Fedor~V. Fomin, Lukasz Kowalik, Daniel Lokshtanov, D{\'{a}}niel
  Marx, Marcin Pilipczuk, Michal Pilipczuk, and Saket Saurabh.
\newblock {\em Parameterized Algorithms}.
\newblock Springer, 2015.

\bibitem{DerbiszKMSSV22}
Jan Derbisz, Lawqueen Kanesh, Jayakrishnan Madathil, Abhishek Sahu, Saket
  Saurabh, and Shaily Verma.
\newblock A polynomial kernel for bipartite permutation vertex deletion.
\newblock {\em Algorithmica}, 84(11):3246--3275, 2022.

\bibitem{Diestel-Book}
Reinhard Diestel.
\newblock {\em Graph Theory, 4th Edition}, volume 173 of {\em Graduate texts in
  mathematics}.
\newblock Springer, 2012.

\bibitem{DowneyFellowsBook13}
Rodney~G. Downey and Michael~R. Fellows.
\newblock {\em Fundamentals of Parameterized Complexity}.
\newblock Texts in Computer Science. Springer, 2013.

\bibitem{Dumas023}
Ma{\"{e}}l Dumas and Anthony Perez.
\newblock An improved kernelization algorithm for trivially perfect editing.
\newblock In Neeldhara Misra and Magnus Wahlstr{\"{o}}m, editors, {\em 18th
  International Symposium on Parameterized and Exact Computation, {IPEC} 2023,
  September 6-8, 2023, Amsterdam, The Netherlands}, volume 285 of {\em LIPIcs},
  pages 15:1--15:17. Schloss Dagstuhl - Leibniz-Zentrum f{\"{u}}r Informatik,
  2023.

\bibitem{EscoffierGM10}
Bruno Escoffier, Laurent Gourv{\`{e}}s, and J{\'{e}}r{\^{o}}me Monnot.
\newblock Complexity and approximation results for the connected vertex cover
  problem in graphs and hypergraphs.
\newblock {\em J. Discrete Algorithms}, 8(1):36--49, 2010.

\bibitem{FominLLSTZ19}
Fedor~V. Fomin, Tien{-}Nam Le, Daniel Lokshtanov, Saket Saurabh, St{\'{e}}phan
  Thomass{\'{e}}, and Meirav Zehavi.
\newblock Subquadratic kernels for implicit 3-hitting set and 3-set packing
  problems.
\newblock {\em {ACM} Trans. Algorithms}, 15(1):13:1--13:44, 2019.

\bibitem{fomin2013polynomial}
Fedor~V Fomin, Saket Saurabh, and Yngve Villanger.
\newblock A polynomial kernel for proper interval vertex deletion.
\newblock {\em SIAM Journal on Discrete Mathematics}, 27(4):1964--1976, 2013.

\bibitem{GanianRS17}
Robert Ganian, M.~S. Ramanujan, and Stefan Szeider.
\newblock Discovering archipelagos of tractability for constraint satisfaction
  and counting.
\newblock {\em {ACM} Trans. Algorithms}, 13(2):29:1--29:32, 2017.

\bibitem{GolovachPL15}
Petr~A. Golovach, Dani{\"{e}}l Paulusma, and Erik~Jan van Leeuwen.
\newblock Induced disjoint paths in claw-free graphs.
\newblock {\em {SIAM} J. Discret. Math.}, 29(1):348--375, 2015.

\bibitem{golumbic2004algorithmic}
Martin~Charles Golumbic.
\newblock {\em Algorithmic graph theory and perfect graphs}, volume~57.
\newblock Elsevier, 2004.

\bibitem{HeggernesKM09}
Pinar Heggernes, Dieter Kratsch, and Daniel Meister.
\newblock Bandwidth of bipartite permutation graphs in polynomial time.
\newblock {\em J. Discrete Algorithms}, 7(4):533--544, 2009.

\bibitem{JacobKMR23}
Ashwin Jacob, Jari J.~H. de~Kroon, Diptapriyo Majumdar, and Venkatesh Raman.
\newblock Deletion to scattered graph classes {I} - case of finite number of
  graph classes.
\newblock {\em J. Comput. Syst. Sci.}, 138:103460, 2023.

\bibitem{JacobMR23jcss}
Ashwin Jacob, Diptapriyo Majumdar, and Venkatesh Raman.
\newblock Deletion to scattered graph classes {II} - improved {FPT} algorithms
  for deletion to pairs of graph classes.
\newblock {\em J. Comput. Syst. Sci.}, 136:280--301, 2023.

\bibitem{JacobMR23algo}
Ashwin Jacob, Diptapriyo Majumdar, and Venkatesh Raman.
\newblock Expansion lemma - variations and applications to polynomial-time
  preprocessing.
\newblock {\em Algorithms}, 16(3):144, 2023.

\bibitem{JacobMZ24}
Ashwin Jacob, Diptapriyo Majumdar, and Meirav Zehavi.
\newblock A polynomial kernel for deletion to the scattered class of cliques
  and trees.
\newblock In Juli{\'{a}}n Mestre and Anthony Wirth, editors, {\em 35th
  International Symposium on Algorithms and Computation, {ISAAC} 2024, December
  8-11, 2024, Sydney, Australia}, volume 322 of {\em LIPIcs}, pages
  41:1--41:17. Schloss Dagstuhl - Leibniz-Zentrum f{\"{u}}r Informatik, 2024.

\bibitem{JacobBDP21}
Hugo Jacob, Thomas Bellitto, Oscar Defrain, and Marcin Pilipczuk.
\newblock Close relatives (of feedback vertex set), revisited.
\newblock In Petr~A. Golovach and Meirav Zehavi, editors, {\em 16th
  International Symposium on Parameterized and Exact Computation, {IPEC} 2021,
  September 8-10, 2021, Lisbon, Portugal}, volume 214 of {\em LIPIcs}, pages
  21:1--21:15. Schloss Dagstuhl - Leibniz-Zentrum f{\"{u}}r Informatik, 2021.

\bibitem{JansenK023}
Bart M.~P. Jansen, Jari J.~H. de~Kroon, and Michal Wlodarczyk.
\newblock Single-exponential {FPT} algorithms for enumerating secluded f-free
  subgraphs and deleting to scattered graph classes.
\newblock In Satoru Iwata and Naonori Kakimura, editors, {\em 34th
  International Symposium on Algorithms and Computation, {ISAAC} 2023, December
  3-6, 2023, Kyoto, Japan}, volume 283 of {\em LIPIcs}, pages 42:1--42:18.
  Schloss Dagstuhl - Leibniz-Zentrum f{\"{u}}r Informatik, 2023.

\bibitem{JansenKW25}
Bart M.~P. Jansen, Jari J.~H. de~Kroon, and Michal Wlodarczyk.
\newblock Single-exponential {FPT} algorithms for enumerating secluded f-free
  subgraphs and deleting to scattered graph classes.
\newblock {\em J. Comput. Syst. Sci.}, 148:103597, 2025.

\bibitem{JohnsonPP20}
Matthew Johnson, Giacomo Paesani, and Dani{\"{e}}l Paulusma.
\newblock Connected vertex cover for
  (sp\({}_{\mbox{1+p\({}_{\mbox{5)}}\)}}\)-free graphs.
\newblock {\em Algorithmica}, 82(1):20--40, 2020.

\bibitem{KlimosovaMMNPS20}
Tereza Klimosov{\'{a}}, Josef Mal{\'{\i}}k, Tom{\'{a}}s Masar{\'{\i}}k, Jana
  Novotn{\'{a}}, Dani{\"{e}}l Paulusma, and Veronika Sl{\'{i}}vov{\'{a}}.
\newblock Colouring (p\({}_{\mbox{r}}\) + p\({}_{\mbox{s}}\))-free graphs.
\newblock {\em Algorithmica}, 82(7):1833--1858, 2020.

\bibitem{KociumakaP14}
Tomasz Kociumaka and Marcin Pilipczuk.
\newblock Faster deterministic feedback vertex set.
\newblock {\em Inf. Process. Lett.}, 114(10):556--560, 2014.

\bibitem{KratschMT08}
Dieter Kratsch, Haiko M{\"{u}}ller, and Ioan Todinca.
\newblock Feedback vertex set on at-free graphs.
\newblock {\em Discret. Appl. Math.}, 156(10):1936--1947, 2008.

\bibitem{Kumabe2025}
Soh Kumabe.
\newblock Quadratic kernel for cliques or trees vertex deletion.
\newblock In Ho-Lin Chen, Wing-Kai Hon, and Meng-Tsung Tsai, editors, {\em 36th
  International Symposium on Algorithms and Computation, {ISAAC} 2025, Tainan,
  Taiwan}, volume CoRR of {\em LIPIcs}, page arXiv:2509.16815. Schloss Dagstuhl
  - Leibniz-Zentrum f{\"{u}}r Informatik, 2025.

\bibitem{lekkeikerker1962representation}
Cornelis Lekkeikerker and Johan Boland.
\newblock Representation of a finite graph by a set of intervals on the real
  line.
\newblock {\em Fundamenta Mathematicae}, 51(1):45--64, 1962.

\bibitem{LewisY80}
John~M. Lewis and Mihalis Yannakakis.
\newblock The node-deletion problem for hereditary properties is np-complete.
\newblock {\em J. Comput. Syst. Sci.}, 20(2):219--230, 1980.

\bibitem{MartinPL20}
Barnaby Martin, Dani{\"{e}}l Paulusma, and Erik~Jan van Leeuwen.
\newblock Disconnected cuts in claw-free graphs.
\newblock {\em J. Comput. Syst. Sci.}, 113:60--75, 2020.

\bibitem{minty1980maximal}
George~J Minty.
\newblock On maximal independent sets of vertices in claw-free graphs.
\newblock {\em Journal of Combinatorial Theory, Series B}, 28(3):284--304,
  1980.

\bibitem{NeidermeierR2003}
Rolf Niedermeier and Peter Rossmanith.
\newblock An efficient fixed-parameter algorithm for 3-hitting set.
\newblock {\em J. Discrete Algorithms}, 1(1):89--102, 2003.

\bibitem{RamanSS06}
Venkatesh Raman, Saket Saurabh, and C.~R. Subramanian.
\newblock Faster fixed parameter tractable algorithms for finding feedback
  vertex sets.
\newblock {\em {ACM} Trans. Algorithms}, 2(3):403--415, 2006.

\bibitem{Thomasse10}
St{\'{e}}phan Thomass{\'{e}}.
\newblock A 4\emph{k}\({}^{\mbox{2}}\) kernel for feedback vertex set.
\newblock {\em {ACM} Trans. Algorithms}, 6(2):32:1--32:8, 2010.

\bibitem{TSUR2025Kernel}
Dekel Tsur.
\newblock Faster algorithms and a smaller kernel for cliques or trees vertex
  deletion.
\newblock {\em Information Processing Letters}, 190:106570, 2025.

\end{thebibliography}

\end{document}